\newif\ifarxiv
\arxivtrue

\ifarxiv{}
\documentclass[USenglish,a4paper,pagebackref,11pt]{article}
\usepackage{amssymb}
\usepackage[sort,numbers]{natbib}
\usepackage{authblk}
\usepackage{a4wide}
\else{}
\documentclass[sigconf]{acmart}
\fi{}

\usepackage{amsmath}
\usepackage{mathtools}
\usepackage{dsfont}
\usepackage{xcolor}
\usepackage{amsthm}
\usepackage{thmtools,thm-restate}

\usepackage{xparse} 
\usepackage{booktabs,array,colortbl,multirow,tabularx}	
\usepackage{graphicx}
\usepackage{enumitem}	
\usepackage{comment}
\usepackage{stmaryrd}
\usepackage{url}
\usepackage{etoolbox}

\usepackage{subcaption}

\usepackage{cleveref}
\crefname{theorem}{Theorem}{Theorems}
\crefname{lemma}{Lemma}{Lemmata}
\crefname{definition}{Definition}{Definitions}
\crefname{corollary}{Corollary}{Corollaries}
\crefname{proposition}{Proposition}{Propositions}
\crefname{equation}{Equation}{Equations}
\crefname{section}{Section}{Sections}
\crefname{figure}{Figure}{Figures}
\crefname{appendix}{Appendix}{Appendices} 
\crefname{algocf}{Algorithm}{Algorithms} 
\Crefname{theorem}{Thm}{Thm}
\Crefname{lemma}{Lemma}{Lemmata}
\Crefname{corollary}{Cor}{Cor}
\Crefname{proposition}{Prop}{Prop}

\newtheorem{corollary}{Corollary}
\newtheorem{lemma}{Lemma}

\theoremstyle{definition}
\newtheorem{definition}{Definition}

\newcolumntype{\ipMR}{>{$\displaystyle\,}r<{$}@{\hspace{0.0em}}}
\newcolumntype{\ipMRspace}{>{$\displaystyle}r<{$}@{\hspace{0.2em}}}

\newcolumntype{\ipMC}{>{$\displaystyle\,}c<{$}@{\hspace{0.0em}}}

\newcolumntype{\ipMLspace}{>{$\displaystyle}l<{$}@{\hspace{0.2em}}}
\newcolumntype{\ipMLmorespace}{>{$\displaystyle}l<{$}@{\hspace{0.3em}}}

\newcounter{IPnumber}
\newcommand{\tagIt}[1]{\refstepcounter{equation}\textnormal{({\theequation})} \label{#1} }

\makeatletter

\newcommand{\eqnum}{\leavevmode\hfill\refstepcounter{equation}\textup{\tagform@{\theequation}}}
\makeatother

\makeatletter
\newcommand{\nosemic}{\renewcommand{\@endalgocfline}{\relax}}
\newcommand{\dosemic}{\renewcommand{\@endalgocfline}{\algocf@endline}}
\let\oldnl\nl
\newcommand{\nonl}{\renewcommand{\nl}{\let\nl\oldnl}}
\makeatother

\makeatletter
\newcommand{\removelatexerror}{\let\@latex@error\@gobble}
\makeatother

\newcounter{ipCounter}
\NewDocumentEnvironment{IPFormulation}{m}{%
\refstepcounter{ipCounter}
\begin{algorithm}[#1]%

}{%
\end{algorithm}
\addtocounter{algocf}{-1}
}

\NewDocumentEnvironment{IPFormulationStar}{m}{%
\refstepcounter{ipCounter}
\begin{algorithm*}[#1]%

}{%
\end{algorithm*}
\addtocounter{algocf}{-1}
}
\makeatletter
\DeclareFontEncoding{LS1}{}{\noaccents@}
\DeclareFontEncoding{LS2}{}{\@noaccents}
\makeatother
\DeclareFontSubstitution{LS1}{stix}{m}{n}
\DeclareFontSubstitution{LS2}{stix}{m}{n}

\DeclareSymbolFont{stixSymbolFont}{LS1}{stixfrak} {m} {n}
\DeclareSymbolFont{largesymbolsstix}{LS2}{stixex}{m}{n}

\DeclareMathDelimiter{\llangle}{\mathopen}{stixSymbolFont}{"28}{stixSymbolFont}{"28}
\DeclareMathDelimiter{\rrangle}{\mathclose}{stixSymbolFont}{"29}{stixSymbolFont}{"29}

\DeclareMathDelimiter{\lbrbrak}{\mathopen}{largesymbolsstix}{"EE}{largesymbolsstix}{"14}
\DeclareMathDelimiter{\rbrbrak}{\mathclose}{largesymbolsstix}{"EF}{largesymbolsstix}{"15}

\DeclareDocumentCommand{\restrict}{d[] d[]}{\ensuremath{\llangle #1 \vert #2 \rrangle}}

\DeclareDocumentCommand{\subLP}{d[] d[]}{\ensuremath{#1 \llbracket #2 \rrbracket}}

\newcommand{\preals}{\ensuremath{\mathbb{R}_{\geq 0}}}
\newcommand{\N}{\mathbb{N}}

\newcommand{\req}{R}

\newcommand{\VG}[1][\req]{\ensuremath{G_{#1}}}
\newcommand{\VV}[1][\req]{\ensuremath{V_{#1}}}
\newcommand{\VE}[1][\req]{\ensuremath{E_{#1}}}

\newcommand{\Vcap}[1][\req]{\ensuremath{d_{#1}}}

\newcommand{\SG}{\ensuremath{G_S}}

\newcommand{\SV}{\ensuremath{V_S}}
\newcommand{\SE}{\ensuremath{E_S}}
\newcommand{\tildeSG}{\ensuremath{\widetilde G_S}}
\newcommand{\tildeSV}{\ensuremath{\widetilde V_S}}
\newcommand{\tildeSE}{\ensuremath{\widetilde E_S}}

\newcommand{\Scap}{\ensuremath{d_{S}}}
\newcommand{\tildeScap}{\ensuremath{\widetilde d_{S}}}

\newcommand{\Scost}{\ensuremath{c_{S}}}
\newcommand{\tildeScost}{\ensuremath{\widetilde c_{S}}}

\newcommand{\map}[1][\req]{\ensuremath{m_{#1}}}
\newcommand{\mapV}[1][\req]{\ensuremath{m^V_{#1}}}
\newcommand{\mapE}[1][\req]{\ensuremath{m^E_{#1}}}
\newcommand{\tildemapV}[1][\req]{\ensuremath{\widetilde{m}^V_{#1}}}
\newcommand{\tildemapE}[1][\req]{\ensuremath{\widetilde{m}^E_{#1}}}

\newcommand{\cut}{\ensuremath{\mathrm{cut}}}

\NewDocumentCommand{\decomp}{O{\req} O{k}}{\ensuremath{D_{#1}^{#2}}}
\NewDocumentCommand{\decompHat}{O{\req} O{k}}{\ensuremath{{\hat{D}}_{#1}^{#2}}}
\NewDocumentCommand{\prob}{O{\req} O{k}}{\ensuremath{f_{#1}^{#2}}}
\NewDocumentCommand{\mapping}{O{\req} O{k}}{\ensuremath{m_{#1}^{#2}}}

\NewDocumentCommand{\PotEmbeddings}{O{\req}}{\ensuremath{\mathcal{D}_{#1}}}
\NewDocumentCommand{\PotEmbeddingsHat}{O{\req}}{\ensuremath{\hat{\mathcal{D}}_{#1}}}

\NewDocumentCommand{\maxDemandV}{O{\req} O{u}}{\ensuremath{{d}_{\textnormal{max}}(#1,#2)}}
\NewDocumentCommand{\maxDemandE}{O{\req} O{u} O{v}}{\ensuremath{d_{\textnormal{max}} ({#1,#2,#3})}}
\NewDocumentCommand{\maxDemandX}{O{\req} O{x}}{\ensuremath{d_{\textnormal{max}} ({#1,#2})}}
\NewDocumentCommand{\maxAllocV}{O{\req} O{u}}{\ensuremath{{A_{\textnormal{max}}}({#1,#2})}}
\NewDocumentCommand{\maxAllocE}{O{\req} O{u} O{v}}{\ensuremath{ {A_{\textnormal{max}}}({#1,#2,#3})}}
\NewDocumentCommand{\maxAllocX}{O{\req} O{x}}{\ensuremath{ {A_{\textnormal{max}}}({#1,#2})}}

\NewDocumentCommand{\alloc}{O{x} O{\map}}{\ensuremath{{A}({#2,#1})}}

\newcommand{\simplePaths}{\ensuremath{\mathcal{P}_S}}

\newcommand{\compP}{\ensuremath{\textsc{P}}}
\newcommand{\compNP}{\ensuremath{\textsc{NP}}}
\newcommand{\NP}{\ensuremath{\textsc{NP}}}
\newcommand{\compPeqNP}{\ensuremath{\compP{\,=\,}\compNP}}

\newcommand{\poly}{\ensuremath{\mathsf{poly}}}

\newcommand{\VNEP}{\textsc{VNEP}}
\newcommand{\minVNEP}{\textsc{Min-VNEP}}

\newcommand{\FPT}{\ensuremath{\textsc{FPT}}}
\newcommand{\Wone}{\ensuremath{\textsc{W[1]}}}
\newcommand{\XP}{\ensuremath{\textsc{XP}}}
\newcommand{\VMP}{\textnormal{\textsc{VMP}}}

\newcommand{\tw}{\ensuremath{\mathrm{tw}}}

\newcommand{\TVbfs}[1][\req]{\ensuremath{V_{T}}}
\newcommand{\TEbfs}[1][\req]{\ensuremath{\overrightarrow{E}_{T}}}

\NewDocumentCommand{\TbagInter}{O{t_1} O{t_2}}{\ensuremath{B_{#1 \cap #2}}}

\NewDocumentCommand{\MappingSpace}{d[]}{\ensuremath{\mathcal{M}(#1)}}

\NewDocumentCommand{\tij}{O{i,j}}{\ensuremath{t_{#1}}}
\NewDocumentCommand{\ti}{O{i}}{\ensuremath{t_{#1}}}

\NewDocumentCommand{\TABentry}{O{t} O{\mapV}}{\ensuremath{\mathbb{T}_{#1}[#2]}}
\NewDocumentCommand{\TABmapping}{O{t} O{\mapV}}{\ensuremath{\mathbb{M}[{#1}][#2]}}
\NewDocumentCommand{\TABMappingCost}{O{t} O{\mapV}}{\ensuremath{\mathbb{C}[{#1}][#2]}}
\NewDocumentCommand{\TABpath}{O{i,j} O{u,v}}{\ensuremath{\mathbb{P}[#1][#2]}}
\NewDocumentCommand{\TABPathCost}{O{i,j} O{u,v}}{\ensuremath{\mathbb{C}_{\mathbb{P}}[{#1}][#2]}}

\usepackage[ruled,vlined,linesnumbered,algosection,hanginginout]{algorithm2e}

\newcommand{\appsymb}{$\bigstar$}
\newcommand{\appref}[1]{\hyperref[proof:#1]{\appsymb}}

\newcommand{\appendixsection}[1]{%
		\ifarxiv{}\else{}%
	\gappto{\appendixProofText}{\label{app:#1}}
	\fi{}%
}

\newcommand{\toappendix}[1]{%
		\ifarxiv{}%
		#1%
		\else{}%
  \gappto{\appendixProofText}
  {{%
    #1%
  }}%
  \fi{}%
}

\newcommand{\appendixproof}[2]{%
		\ifarxiv{}%
				#2%
\else{}%
  \gappto{\appendixProofText}
  {
    \subsection{Proof of \autoref{#1}}\label{proof:#1}
    #2
  }
		\fi{}%
}

\ifarxiv{}\else{}
\copyrightyear{2021}
\acmYear{2021}
\setcopyright{acmcopyright}
\acmConference[SPAA '21] {Proceedings of the 33rd ACM Symposium on Parallelism in Algorithms and Architectures}{July 6--8, 2021}{Virtual Event, USA.}
\acmBooktitle{Proceedings of the 33rd ACM Symposium on Parallelism in Algorithms and Architectures (SPAA '21), July 6--8, 2021, Virtual Event, USA}
\acmPrice{15.00}
\acmISBN{978-1-4503-8070-6/21/07}
\acmDOI{10.1145/3409964.3461787}

\fi{}

\begin{document}

\newcommand{\thetitle}{Optimal Virtual Network Embeddings for Tree Topologies}

\ifarxiv{}

\title{\Large\bf\thetitle\footnote{An extended abstract of this work appears in the Proceedings of the 33rd ACM Symposium on Parallelism in Algorithms and Architectures (SPAA '21).}}

\author[1]{Aleksander~Figiel\thanks{Supported by DFG, Project MaMu NI 369/19. E-Mail: \texttt{aleksander.figiel@campus.tu-berlin.de}}}
\author[1]{Leon~Kellerhals\thanks{E-Mail: \texttt{leon.kellerhals@tu-berlin.de}}}
\author[1]{Rolf~Niedermeier\thanks{E-Mail: \texttt{rolf.niedermeier@tu-berlin.de}}}
\author[2,3]{Matthias~Rost\thanks{E-Mail: \texttt{matthias.johannes.rost@sap.com}}}
\author[2,4]{Stefan~Schmid\thanks{Funded by European Research Council (ERC), grant agreement 864228 (AdjustNet).\\ \phantom{xyz}~E-Mail: \texttt{stefan\_schmid@univie.ac.at}}}
\author[1]{Philipp~Zschoche\thanks{E-Mail: \texttt{zschoche@tu-berlin.de}}}
\affil[1]{\small Technische Universität Berlin, Algorithmics~and~Computational~Complexity, Germany}
\affil[2]{\small Technische Universität Berlin, Data~Communications~and~Networking, Germany}
\affil[3]{\small SAP SE, Germany}
\affil[4]{\small University of Vienna, Austria}

\date{\vspace{-1cm}}

\fi{}

\ifarxiv{}
	\maketitle
\fi{}

\begin{abstract}
The performance of distributed and data-centric applications
often critically depends on the interconnecting network.
Applications are hence modeled as virtual networks, also accounting for resource demands on links.
At the heart of provisioning such virtual networks lies the NP-hard Virtual Network Embedding Problem~(VNEP): how to jointly map the virtual nodes and links onto a physical substrate network at minimum cost while obeying capacities.

This paper studies the VNEP in the light of parameterized complexity. We focus on tree topology substrates, a case often encountered in practice and for which the VNEP remains NP-hard.
We provide the first fixed-parameter algorithm for the VNEP with running time $O(3^r (s+r^2))$ for requests and substrates of $r$ and $s$ nodes, respectively.
In a computational study our algorithm yields running time improvements in excess of 200$\times$ compared to state-of-the-art integer programming approaches.
This makes it comparable in speed to the well-established ViNE heuristic while providing optimal solutions.
We complement our algorithmic study with hardness results for the VNEP and related problems.
\end{abstract}

\ifarxiv{}\fi

\ifarxiv{}\fi{}


\section{Introduction}\label{sec:intro}

Data-centric and distributed applications, 
including batch processing, streaming, scale-out data\-bases, 
or distributed machine learning, generate a significant amount of
network traffic and their performance critically depends
on the underlying network. 
As the network infrastructure is often shared and the bandwidth available 
can vary significantly over time, this can have a
non-negligible impact on the application performance \cite{talkabout}. 

Network virtualization has emerged as a promising solution to 
ensure a predictable application performance over shared infrastructures, 
by providing a virtual network abstraction which comes with explicit bandwidth 
guarantees~\cite{vnep}. 
In a nutshell,
a virtual network request is modeled as a directed graph $\VG=(\VV,\VE)$ whose elements are attributed with resource demands.
The nodes represent, e.g., containers or virtual machines, requesting, e.g., CPU cores and memory, while the edges represent communication channels of a certain bandwidth.
Formally, the demands of a virtual network request~$\VG$ are a
function~$\Vcap: \VG \to \preals^{\tau}$, $\tau \in \N$, of every node and every edge onto a~$\tau$-dimensional vector of nonnegative reals.
To provision such a virtual network request in a (shared) physical substrate network, also modeled as a directed graph~$\SG = (\SV, \SE)$ with capacities~$\Scap: \SG \to \preals^{\tau}$, we need to find an \emph{embedding} that maps the request nodes onto the substrate nodes and the request edges onto paths in the substrate while respecting capacities.


The NP-hard \textsc{Virtual Network Embedding Problem}, asking to find such embeddings, poses the main challenge of provisioning virtual networks and has been studied for various objectives~\cite{vnep-survey}.
In this paper, we study the following central cost-minimization variant (see \cref{def:vnep} for a formal definition):

\smallskip
\noindent\textsc{Minimum-Cost Virtual Network Embedding} (\minVNEP)
\begin{description}[labelwidth=1cm, leftmargin=1.1cm]
	\item[Input:]	A directed graph~$\SG = (\SV, \SE)$ on~$s$ nodes, called \emph{substrate}, and a directed graph $\VG = (\VV, \VE)$ on~$r$ nodes, called \emph{request}, with~$\tau$-dimensional demands~$\Vcap: \VG \to \preals^\tau$, capacities~$\Scap: \SG \to \preals^\tau$, and costs~$\Scost: \SG \to \preals^\tau$.
	\item[Task:] 
		Find mappings of the request onto substrate nodes and of the request edges onto paths in the substrate, such that
		\begin{enumerate}[label=(\arabic*), leftmargin=*]
			\item the node and edge capacities are respected by the node and edge mappings, and
			\item the cost of all nodes and edges used by the mapping is minimized.
		\end{enumerate}
%
%
%
\end{description}	

We remark that several other variants of the \textsc{Virtual Network Embedding Problem} can be reduced to \minVNEP{}  (see Section~\ref{sec:rel-work-and-novelty}).

%
%
%


\subsection{Contributions and Techniques}
\label{sec:contributions}
While the \minVNEP{} is known to be notoriously hard in general~\cite{ton20hard},
real-world network optimization problems often exhibit a specific structure.
In this work, we provide efficient, exact algorithms that exploit such a structural property.
Our main theoretical contribution is a fixed-parameter algorithm for the \minVNEP{} onto tree substrates when parameterized by the number of nodes in the request, that is, we present an algorithm which performs very well for small request graphs:

\begin{restatable}{theorem}{thmSimpleDP}
	\label{thm:simple-dp}
	$\minVNEP$ can be solved in $\mathcal O(3^{r}(s + r^2))$ time when the substrate~$\SG$ is a tree, where~$r = |\VV|$ and~$s = |\SV|$.
\end{restatable}

From a theoretical (worst-case) point of view there is almost no hope to obtain a substantially faster (exact) algorithm for tree substrates (see~Section~\ref{sec:hardness}).
A specific feature of the algorithm is its robustness:
It can be easily modified to also support additional constraints such as mapping exclusions on a per-node or per-edge basis~\cite{ccr19tw}.
Furthermore, as a side result, we show that any instance of \minVNEP{} on tree substrates can be translated in linear time into an instance of \minVNEP{} in which the substrate is a binary tree and only its leaves have non-zero capacities.
Hence, algorithms designed for such tree substrates, as, e.g., those by \citet{oktopus} and \citet{ccr15emb}, can also be applied on general tree substrates.

The algorithm of \cref{thm:simple-dp} also performs very well in practice.
In an extensive computational study we compare our algorithm to the classical exact algorithm based on integer programming as well as to the well-established ViNE heuristic~\cite{vnep}.  
The results are clear: our algorithm outperforms the integer program on all instances, consistently yielding average speedups exceeding a factor of $100\times$ and often even a factor of~$200\times$ for densely connected request graphs across small to medium-sized data center networks.
The running time of ViNE lies in the same order of magnitude as the one of our algorithm, but produces feasible solutions only for a quarter of the instances for which our algorithm found an \emph{optimal} solution.
To ensure reproducibility and facilitate follow-up work, we will provide 
our implementation to the research community as open source code, together with
all experimental artefacts.

As mentioned before, we complement our algorithm (\cref{thm:simple-dp}) by showing that in theory there is little hope for improving its running time substantially.
This can be derived from a simple NP-hardness result for the decision version of \minVNEP{}, which we will call \VNEP{}.
Here, we are given an instance of \minVNEP{} together with an integer $k$ and ask whether there is an embedding with costs at most $k$.
We show the following.

\begin{restatable}{theorem}{thmNPHardness}
	\label{obs:np-hard}
	\VNEP{} is \NP-hard, even if the subtrate $\SG$ consists of two nodes and the request~$\VG$ is edgeless, and $k=0$.
\end{restatable}

An intermediate question from \cref{thm:simple-dp} is whether we can find another graph parameter $x$ of the request which is asymptotically smaller
than $r$ (number of vertices) but still admits an exact algorithm of running time $f(x) (s+r)^{O(1)}$, where $f$ is a computable function.
Assuming P$\ne$NP, such a running time cannot be achieved for any parameter $x$ which is asymptotically smaller than 
the number of edges in the request.
This is because the NP-hardness for the \VNEP{} holds even if the request contains no edges.
%
%
%
Also, \cref{obs:np-hard} rules out the existence of any approximation algorithm for the \minVNEP{}, even if the degree of the polynomial may depend on the substrate's number of nodes and the request's number of edges.

Our last contribution is a conditional lower bound on the running time of the \emph{Valid Mapping Problem} (\VMP{}), a relaxation of the \VNEP{}:
Analogously to the \VNEP{}, the question is whether there are node and edge mappings of the request onto the substrate such that the cost is below a given~$k \in \preals$, but we only enforce that the mapping of each individual virtual element does not exceed the capacities of the substrate (see \Cref{sec:hardness} for a formal definition).
This relaxation is used for instance by \citet{ccr19tw} to obtain an approximation algorithm for \VNEP{} in the resource augmentation framework.
Specifically, they present an algorithm for \VMP{} running in~$\poly(r) \cdot s^{\mathcal{O}{(\tw(\VG))}}$ time, where~$s$ and~$r$ are the number of nodes in the substrate and the request, respectively, and~$\tw(\VG)$ is the treewidth of the request~\cite{flum2006parameterized}.
By proving a W[1]-hardness result, we show that there is presumably no fixed-parameter algorithm for \VMP{} parameterized by the cost upper bound~$k$ combined with the number of nodes~$r$ in the request, and that the running time for \VMP{} obtained by \citet{ccr19tw} is asymptotically optimal:

\ifarxiv{}
\begin{restatable}{theorem}{thmETHBound}
	\label{thm:eth}
\VMP{} parameterized by~$k + r$ is W[1]-hard and, unless the Exponential Time Hypothesis (ETH) fails, there is no algorithm for \VMP{} running in~$f(r) \cdot s^{o(r)}$ time, where~$r$ and~$s$ are the number of nodes in the request and the substrate, respectively.
\end{restatable}
\fi{}

\subsection{Related Work and Novelty}
\label{sec:rel-work-and-novelty}


The \textsc{Virtual Network Embedding Problem} has received tremendous attention by the networking community over the last 15~years: already by 2013 more than 80 algorithms had been published in the literature for its various flavors~\cite{vnep-survey}. The particular $\minVNEP$ objective, on which we focus in this paper, has received by far the most attention: there is extensive work on heuristics~\cite{zhu2006algorithms, Lischka, vnep, meng2010improving} as well as exact algorithms based on mixed-integer programs~\cite{vnep, infuhr2011introducing} for~$\minVNEP$.
Notably, however, there is no work so far on (nontrivial) combinatorial \emph{exact} algorithms for~$\minVNEP$.

\paragraph{Closely Related Applications.}
Various applications of the \textsc{Virtual Network Embedding Problem} have spawned independent research with dozens of proposed algorithms.
Among the most prominent ones are the embedding problems pertaining to Virtual Clusters (VCEP)~\cite{oktopus}, to Service Function Chains (SFCEP)~\cite{nfv-sfc-survey}, to Virtual Data Centers (VDCEP)~\cite{vdc-Zhang2014Infocom}, and to the Internet of Things~\cite{SHN20}.
In short, the VCEP studies the embedding of tree requests onto data center topologies, the SFCEP studies the embedding of sparse requests representing (virtualized) network functions, and the VDCEP focuses on the embedding of arbitrary requests across geographically distributed data centers in wide-area networks.
While at times introducing additional constraints, the Virtual Network Embedding Problem lies at the heart of these problems as well.

\paragraph{Applications of $\minVNEP$.}
Various algorithms rely on solving the $\minVNEP$ as a subroutine.
The application domains include:
\begin{description}
	\item[Offline Objectives.] The offline setting of the \textsc{Virtual Network Embedding Problem} 
	over several requests under cost objectives can be solved by~$\minVNEP$ by considering the union of the requests.
	Further, there are exponential-time (parameterized) approximations for the offline setting in the resource augmentation framework that use algorithms for the cost minimization variant of~$\VMP$ or~$\minVNEP$ as a subroutine~\cite{,rostSchmidVNEP_RR_IFIP_18,ccr19tw,ifip20vnep}.
	\item[Competitive online optimization.] \mbox{\citet{tcs12vnet}} showed how to construct competitive online algorithms for the profit variant of the VNEP from any exact algorithm for the~$\minVNEP$.
	\item[Congestion minimization.] \citet{bansal2011minimum} studied the problem of minimizing the maximal load (while not enforcing capacities).
	They obtained competitive online and offline approximation algorithms that solve~$\minVNEP$ as a subroutine.
\end{description}

Given our fixed-parameter algorithm for the special case of tree substrates (see Theorem~\ref{thm:simple-dp}), novel parameterized algorithms for all of the above highlighted settings and objectives can be obtained.

\paragraph{(Parameterized) Complexity.}
Despite the popularity of the $\VNEP$, until recently only little was known about its 
fine-grained computational complexity. 
\citet{ton20hard} made the first step towards understanding the (parameterized) complexity of the $\VNEP$, 
showing that any optimization variant of the $\VNEP$ (where \minVNEP{} is one of them) is inapproximable in polynomial time, 
unless $\compPeqNP$, even when the request graph is planar and the substrate is acyclic. 
\citet{ccr19tw} gave the first approximation algorithm for the offline profit objective for requests of constant treewidth in the resource augmentation framework, also carrying over to the cost setting~\cite{ifip20vnep}. 

In contrast to the above works, we focus on efficient and exact fixed-parameter algorithms while restricting the substrate to be a tree.
Tree substrates are most predominantly encountered in data centers, e.g., in the form of fat trees~\cite{fattrees}.
Fat trees or similar leaf-spine architectures are widely studied in the literature and used in industry~\cite{vl2,oktopus}.
Additionally, by employing substrate transformations, such as computing Gomory-Hu trees~\cite{soualah-vnep-via-gomory-hu-noms-2016}, non-tree substrates may be transformed to trees, albeit optimality guarantees cannot be preserved.
\citet{bansal2011minimum} designed specific algorithms for tree substrates of bounded depth, where the objective is to minimize congestion. 
For the parameterization of the request size---the main focus of this paper---no results are known thus far.

\paragraph{Small Request Graphs.} 
The application of our main result (cf.~\cref{thm:simple-dp}) yields algorithms of practical significance only when the number of request nodes is small and in our computational study we restrict our attention to request graphs on less than 12 nodes. While this may be considered to be an unreasonably small number of nodes, many existing works on the VNEP~\cite{vnep, vnep-survey} and its applications in data centers~\cite{Yuan-VNEP-DC-2019-ACCESS, Sun-SFC-DC-2019-ACCESS} consider requests of such size.


\subsection{Preliminaries}
For~$n \in \N$ let~$[n] := \{1, \dots, n\}$.
For two vectors~$a = (a_i)_{i=1}^{\tau}, b = (b_i)_{i=1}^{\tau}$ we write~$a \le b$ if~$a_i \le b_i$ for all~$i \in [\tau]$ and~$a \not\le b$ otherwise.

Let~$G = (V, E)$ be a directed graph.
For a node subset~$V' \subseteq V$, we denote by~$G[V']$ the \emph{subgraph of~$G$ induced by~$V'$}, and by~$V(G[V'])$ and~$E(G[V'])$ the node set and the edge set of~$G[V']$, respectively.
For a node~$v \in V$ we denote by~$N_G^+(v)$, respectively~$N_G^-(v)$, the set of nodes that are connected by an edge pointing away from, respectively towards~$v$.
By~$N_G(v) := N_G^+(v) \cup N_G^-(v)$ we denote the (combined) neighborhood of~$v$.
The degree~$\deg_G(v)$ of~$v$ is the number of nodes in the neighborhood of~$v$.
The \emph{underlying undirected graph} of a directed graph~$G$ is the undirected graph without multiedges on the same node set and it has an edge~$\{u, v\}$ for every directed edge~$(u, v)$ in~$G$.
We say that a directed graph is a \emph{tree} if its underlying undirected graph is a tree.

%
%
Given an instance of either \minVNEP{}, its decision variant \VNEP{} or the \VMP{}, we say that a pair of mappings~$(\mapV, \mapE)$ is a \emph{valid mapping} if the edge mappings are valid, and capacities are respected per each individual virtual element, that is,
\begin{enumerate}
		\item
				for every edge~$(u, v) \in E_r$, $\mapE(u, v)$ is a path from~$\mapV(u)$ to~$\mapV(v)$ in~$\SG$,
		\item
				$\Vcap(w) \le \Scap(\mapV(w))$ for every~$w \in \VV$, and
		\item 	$\Vcap(e) \le \Scap(e_S)$ for all virtual edges~$e \in \VE$ and their mappings~$e_S \in \mapE(e)$.
\end{enumerate}
We call the mapping \emph{feasible} if additionally all demands of the request nodes and edges can be fulfilled by the capacities of the substrate nodes and edges they are mapped onto, that is,
\begin{align*}
\ifarxiv{}
	&\sum_{w : \mapV(w)=v} \Vcap(w) \le \Scap(v) &&\text{ for } v \in \SV, &\text{ and }\\
	&\sum_{e_R:e_S \in \mapE(e_R)} \Vcap(e_R) \le \Scap(e_S) &&\text{ for } e_S \in \SE.&
\fi{}
\end{align*}

The \emph{cost} of a mapping~$(\mapV, \mapE)$ is defined as the sum of the cost of mapping all nodes plus the sum of the costs mapping all edges.
Note that the latter consists of the cost of every single edge of the path onto which a request edge is mapped.
Formally, the cost is
\[
	\sum_{v \in V(G_R)} \Vcap(v)^\top  \Scost(\mapV(v)) 
	+ \sum_{e \in E(G_R)} \Big(\sum_{e' \in \mapE(e)} \Vcap(e)^\top \Scost(e') \Big).
\]

We can now formally define~$\minVNEP$:

\begin{definition}[Min.~Cost Virtual Network Embedding (\minVNEP)]
	\label{def:vnep}~\\[-22pt]
\begin{description}[labelwidth=1cm, leftmargin=1.1cm]
	\item[Input:]	A directed graph~$\SG = (\SV, \SE)$ on~$s$ nodes, called \emph{substrate}, and a directed graph $\VG = (\VV, \VE)$ on~$r$ nodes, called \emph{request}, with demands~$\Vcap: \VG \to \preals^\tau$, capacities~$\Scap: \SG \to \preals^\tau$, and costs~$\Scost: \SG \to \preals^\tau$.
	\item[Task:] Find a feasible mapping of minimum cost.
\end{description}	
\end{definition}

In the decision variant, $\VNEP$, we are additionally given a nonnegative~$k\in\preals$ with an instance of $\minVNEP$ and decide wheth\-er there is a feasible mapping with cost at most~$k$.
Formally, it is defined as follows (note that in this definition we replace the~\mbox{$\tau$-dimensional} vectors by scalars):
\begin{definition}[Virtual Network Embedding Problem (\VNEP)]
	~\\[-12pt]
\begin{description}[labelwidth=1cm, leftmargin=1.1cm]
	\item[Input:]	A directed graph~$\SG = (\SV, \SE)$ on~$s$ nodes, called \emph{substrate}, and a directed graph $\VG = (\VV, \VE)$ on~$r$ nodes, called \emph{request}, with demands~$\Vcap: \VG \to \preals$, capacities~$\Scap: \SG \to \preals$, costs~$\Scost: \SG \to \preals$, and a cost upper bound~$k \in \preals$.
	\item[Question:] Is there a feasible mapping of cost at most~$k$?
\end{description}	
\end{definition}
The Valid Mapping Problem ($\VMP$) takes the same input as the $\VNEP$ and asks whether there is a \emph{valid} (but not necessarily feasible) mapping with cost at most~$k$.

We assume familiarity with standard notions regarding algorithms and complexity, but briefly review notions regarding parameterized complexity analysis.
Let~$\Sigma$ denote a finite alphabet.
A parameterized problem~$L\subseteq \{(x,k)\in \Sigma^*\times \N_0\}$ is a subset of all instances~$(x,k)$ from~$\Sigma^*\times \N_0$,
where~$k$ denotes the parameter.
A parameterized problem~$L$ is 
		\emph{fixed-parameter tractable} (or contained in the class \FPT) 
		 if there is an algorithm that decides every instance~$(x,k)$ for~$L$ in~$f(k)\cdot |x|^{O(1)}$ time, and it is
		 \emph{contained in the class~\XP{}} 
		 if there is an algorithm that decides every instance~$(x,k)$ for~$L$ in~$|x|^{f(k)}$ time,
where~$f$ is any computable function only depending on the parameter and $|x|$ is the size of~$x$.
For two parameterized problems~$L,L'$,
an instance~$(x,k)\in \Sigma^*\times\N_0$ of~$L$ is equivalent to an instance~$(x',k')\in \Sigma^*\times\N_0$ for~$L'$ if~$(x,k)\in L\!\!\iff\!\! (x',k')\in L'$.
A problem~$L$ is \Wone-hard if for every problem~$L'\in \Wone$ there is an algorithm that maps any instance~$(x,k)$
in~$f(k)\cdot |x|^{O(1)}$ time to an equivalent instance~$(x',k')$ with~$k'=g(k)$ for some computable functions~$f,g$.
It holds true that~$\FPT \subseteq \Wone \subseteq \XP$.
It is believed that~$\FPT\neq \Wone$,
and that hence no~\Wone-hard problem is believed to be fixed-parameter tractable.
Another prominent assumption in the literature is the
\emph{Exponential Time Hypothesis} (ETH) which states 
that there is no $2^{o(n)}$-time algorithm for \textsc{$3$-SAT}, where~$n$ is the number of variables~\cite{IPZ01}.

\section{Hardness}
\label{sec:hardness}
\appendixsection{sec:hardness}

In this section, we show that there is no \XP-algorithm to solve optimally, or approximate the costs of, \minVNEP{} for any combined parameter consisting of
(i) any parameter of the substrate and
(ii)~the number of edges in the request, unless P$=$NP.
In related work, we can find several special cases in which \minVNEP{} remains \NP-hard~\cite{amaldi2016computational,bansal2011minimum}.
However, from the parameterized point of view the following simple polynomial-time many-one reduction from \textsc{Partition} to \VNEP{} (the decision version of \minVNEP{}) excludes many potential parameters towards an \FPT- or even an \XP-algorithm.
\thmNPHardness*
\begin{proof}
We reduce from the \NP-hard \textsc{Partition} problem, where we are given 
a multiset $S$ of positive integers
and ask whether
there is a $S'\subseteq S$ such that $\sum_{x \in S'} x = \sum_{x \in S \setminus S'} x$ \cite{karp1972reducibility}.

Let $S$ be such a multiset of positive integers and assume without loss of generality
that $B := \sum_{x \in S} x$ is even.
We construct an instance $I=(G_S,G_R,d_R,d_S,c_S,k=0)$ of \VNEP{} such that
				$G_S := (\{a,b\},\{(a, b), (b, a)\})$,
				$G_R := (S, \emptyset)$ and
				$d_R(x) := x$ for all $x \in S$,
				$c_S(a) := c_S(b) :=  c_S(a, b) := c_S(b, a) := 0$,
				$d_S(a,b) := d_S(b, a) := 0$, and 
				$d_S(a) := d_S(b) := \frac{B}{2}$.
Clearly, this is doable in polynomial time.

We now show that there exists a solution~$S' \subseteq S$ if and only if there exists a feasible mapping~$(\mapV, \mapE)$ for~$I$ of cost~$0$.

($\Rightarrow$): 
Let $S'\subseteq S$ such that $\sum_{x \in S'} x = \sum_{x \in S \setminus S'} x = \frac{B}{2}$.
Then, we set $\mapV(x) = a$, for all $x \in S'$, and
$\mapV(x) = b$, for all $x \in S\setminus S'$.
Observe that $(\mapV,\mapE)$ is a feasible mapping of cost $0$.

($\Leftarrow$):
Let $(\mapV, \mapE)$ be a feasible mapping for~$I$ of cost~$0$.
Let $S'\subseteq S$ be the set of nodes of $G_R$ which are mapped to $a$.
Hence, $\sum_{x \in S'} x \leq d_S(a) = \frac{B}{2}$ 
and $\sum_{x \in S \setminus S'} x \leq d_S(b) = \frac{B}{2}$.
Since $\sum_{x \in S} x = B$, 
we have $\sum_{x \in S'} x = \sum_{x \in S \setminus S'} x$.
\end{proof}

Since \VNEP{} is \NP-hard even if the substrate is of constant size, we can conclude that there is no \XP-algorithm for \VNEP{} parameterized by any reasonable parameter of the substrate, unless P$=$NP.
Otherwise, this would imply a polynomial-time algorithm for the NP-hard \textsc{Partition} problem. 
Furthermore, since \VNEP{} is \NP-hard even if the substrate graph 
is of constant size and the request is edgeless,
we can exclude the existence of an \XP-algorithm for \VNEP{} 
parameterized by a combination of any `reasonable' parameter for the substrate 
and the number of edges in the request.
Note that this excludes among others the parameters \emph{vertex cover number}, \emph{feedback edge number}, \emph{treewidth}, and \emph{maximum degree} of the request,
because these parameters are upper-bounded by the number of edges.
Moreover, since $k=0$ in \cref{obs:np-hard},
any approximation algorithm\footnote{That is, an algorithm returning a feasible solution and giving provable guarantees on the distance of the returned solution to the optimal one.} for~$\minVNEP$ would be able to solve \textsc{Partition}.
Altogether, we have the following.
\begin{corollary}
		\label{cor:hardness}
		Let $f \colon \mathcal G \rightarrow \mathbb N$ be a computable function, where $\mathcal G$ is the set of directed graphs.
		Unless~P$=$NP, 
		\begin{enumerate}
				\item there is no $|I|^{h\left(f(G_S)+|E_R|\right)}$-time algorithm for \VNEP{}, and
				\item there is no $|I|^{h\left(f(G_S)+|E_R|\right)}$-time approximation algorithm for $\minVNEP$,
		\end{enumerate}
		where $|I|$ is the size of the instance, 
		$G_S$ is the substrate,
		$|E_R|$ is the number of edges in the request, 
		and~$h \colon \mathbb N \rightarrow \mathbb N$ is a computable function.
\end{corollary}

Given the hardness results of \cref{cor:hardness}, we see two ways to develop efficient exact algorithms:
\begin{enumerate}
		\item Restrict the input instances to special cases which are relevant in practice---this is what we do in \Cref{sec:robust-dp}. 
		\item Study a reasonable relaxation of the problem---such as the (NP-hard) \VMP.
\end{enumerate}
Towards (2), \citet{ccr19tw} studied and presented an algorithm for the \VMP{} running in~$\poly(r) \cdot s^{\mathcal O(\tw(\VG))}$ time, where~$\tw(\VG)$ is the treewidth of the request.
They then used this algorithm as a subroutine in an approximation algorithm for an offline variant of the Virtual Network Embedding Problem (see Section~\ref{sec:rel-work-and-novelty}).

With \cref{thm:eth}, we show that the algorithm of \citet{ccr19tw} is asymptotically optimal, unless the Exponential Time Hypothesis fails.%
\ifarxiv{}\fi{}
\appendixproof{thm:eth}{%
	For the sake of completeness, we explicitly define the \textsc{Valid Mapping Problem} 
	and show afterwards the formal proof of \cref{thm:eth}.

\begin{definition}[\textsc{Valid Mapping Problem} (\VMP)]
	\label{def:vmp}~\\[-12pt]
\begin{description}[labelwidth=1cm, leftmargin=1.1cm]
	\item[Input:]	
			A directed graph~$\SG = (\SV, \SE)$ called the \emph{substrate graph},
			a directed graph~$\VG = (\VV, \VE)$ called the \emph{request graph}, with demands $\Vcap \colon \VG \to \preals$, a capacities~$\Scap \colon \SG \to \preals$, a costs~$\Scost \colon \SG \to \preals$, and a cost upper-bound~$k \in \preals$.
	\item[Question:]
			Are there mappings~$\mapV \colon \VV \to \SV$ and $\mapE \colon \VE \to \simplePaths$
	such that
	\begin{enumerate}[label=(\arabic*)]
		\item $\Vcap(v) \le \Scap(\mapV(v))$ holds for all~$v \in \VV$, 
		\item for every edge~$(u, v) = e \in \VE$, it holds that $\mapE(e)$ is a path from $\mapV(u)$ to~$\mapV(v)$ and for every edge~$e' \in E(\mapE(e))$, it holds that $\Vcap(e) \le \Scap(e')$, and
\ifarxiv{}
		\item the overall mapping cost 
				\begin{align*}
						\sum_{v \in \VV} \Scost(\mapV(v)) \cdot \Vcap(v) + \sum_{e \in \VE} \big(\sum_{e' \in E(\mapE(e))} \Scost(e') \big) \cdot \Vcap(e)
				\end{align*}
						is at most~$k$?
\fi{}
	\end{enumerate}
\end{description}
\end{definition}
\thmETHBound*

\begin{proof}
	We provide a polynomial-time many-one reduction from the W[1]-hard \cite{fellows2009parameterized} \textsc{Multicolored Clique} problem:
	Given an integer~$k$ and a~$k$-partite undirected graph~$G = (V_1, V_2, \dots, V_k, E)$, \textsc{Multicolored Clique} asks whether~$G$ contains a clique on~$k$ nodes.
	Assuming ETH, there is no~$f(k) \cdot |V(G)|^{o(k)}$-time algorithm for \textsc{Multicolored Clique} \cite{CHKX06}.

	We construct an instance of~\VMP{} as follows:
	We set~$\SV := V(G)$, and for every undirected edge~$\{w_i, w_j\}$, where~$i<j$ for~$w_i \in V_i$ and~$w_j \in V_j$, we add a directed edge~$(w_i, w_j)$ to the edge set~$\SE$ of the substrate graph.
	Our request graph~$\VG := (\{v_1, v_2, \dots, v_k\},\allowbreak \{ (v_i, v_j) \mid 1 \le i < j \le k \})$ is a directed clique.
	For all~$e \in \VE$, we set~$\Vcap(e) := 1$.
	For~$1 \le i \le k$, we set~$\Vcap(v_i) := i+1$.
	For all~$e \in \SE$, we set~$\Scap(e) := 1$.
	For~$1 \le i \le k$ and for~$w \in V_i$, we set~$\Scap(w) := i+1$.
	The cost~$\Scost$ is~$1$ for every edge in~$\SE$ and~$\Scost$ is~$i+1$ for every node in~$\SV$.
	Finally, we set the cost upper bound to~$k' := \sum_{i=1}^k (i+1)^2 + |\VE|$.
	Note that $k' + r \in O(k^3)$.

	We now show that~$(G, k)$ is a yes-instance of \textsc{Multicolored Clique} if and only if the instance of~\VMP{} above is a yes-instance.

	($\Rightarrow$): 
	Let~$G'$ be the multicolored clique in~$G$.
	Then we construct the mapping~$\map = (\mapV, \mapE)$ such that
	\begin{enumerate}[label=(\arabic*)]
		\item for every node~$v_i \in \VV$, we set~$\mapV(v_i)$ to be the (unique) node in~$V(G') \cap V_i$,
		\item for every edge~$(v_i, v_j) \in \VE$, we set~$\mapE(v_i, v_j)$ to be the set of directed edges~$(u_i, u_j) \in \SE$ with~$u_i \in V(G') \cap V_i$ and~$u_j \in V(G') \cap V_j$.
	\end{enumerate}
	The mapping~$\map$ is valid:
	The demands of a node~$v_i$ are equal to the capacity and costs of~$\mapV(v_i)$.
	The resulting costs are~$(i+1)^2$ for each~$v_i \in \VV$.
	For every edge in~$\VE$ there is a path of length one.
	Thus the cost incurred by the mapping is exactly~$k'$.

	($\Leftarrow$): 
	Assume towards a contradiction that there is no clique of size~$k$ in~$G$, but there exists a valid mapping~$\map$ with the costs being at most~$k'$.
	Observe first that, due to the demands and capacities, the nodes~$\VV$ must incur cost of at least~$\sum_{i=1}^k (i+1)^2$.

	Suppose the cost of the nodes are exactly~$\sum_{i=1}^k (i+1)^2$, that is, node~$v_i$ is mapped onto a node in~$V_i$.
	Then the cost of the mapping of the request edges~$\VE$ must be greater than~$|\VE|$ since
	\begin{enumerate}[label=(\arabic*)]
		\item every edge in~$\VE$ is mapped onto a path of length~$\ell\ge1$ 
		\item at least one edge in~$\VE$ is mapped onto a path of length at least two, as~$G$ does not contain a clique on~$k$ nodes.
	\end{enumerate}
	This is a contradiction to the costs of~$\map$ being at most~$k'$.

	So suppose that the cost of the nodes are greater than~$\sum_{i=1}^k (i+1)^2$.
	Since the overall cost of the mapping is at most~$k'$, there must be edges in~$\VE$ that are mapped onto paths of length zero.
	Let~$v_i \in \VV$, and let~$x_i$ be the number of edges leaving~$v_i$ that are mapped onto paths of length zero.
	Then~$v_i$ is mapped onto a node in~$V_h$, where~$h \ge i + x_i$.
	So the mapping of~$v_i$ incurs cost of at least~$(i+1)(i+1+x_i)$, and the mapping of the edges leaving~$v_i$ incur cost of at least~$|N^+(v_i)| - x_i$.
	The overall cost of the mapping~$\map$ thus is~$\sum_{i=1}^k (i+1)^2 + |E_r| + \sum_{i=1}^k i \cdot x_i$, where the last sum accumulates the cost of the edges that are mapped onto a path of length zero.
	This again is a contradiction to the costs of~$\map$ being at most~$k'$.

	Assume now that there is an algorithm for \VMP{} running in~$f(r) \cdot |\SV|^{o(r)}$ time.
	Then we can solve an instance~$(G, k)$ of \textsc{Multicolored Clique} as follows.
	Construct the corresponding \VMP{}-instance in~$n^{\mathcal O(1)}$ time, and solve it in~$f(k) \cdot n^{o(k)}$ time.
	An algorithm for \textsc{Multicolored Clique} with this running time contradicts the ETH.
\end{proof}}

\section{Efficient VNEP algorithm for small requests on trees}
\label{sec:robust-dp}
\appendixsection{sec:robust-dp}

We focus on the special case of \VNEP{} where the substrate is a tree and 
show that it is fixed-parameter tractable 
when parameterized by the number of nodes in the request.
Thus, the main objective of this section is to show the following.
\thmSimpleDP*
Recall that \VNEP{} (and thus \minVNEP{}) on tree substrates is \NP-hard (\cref{obs:np-hard}),
even if the request contains no edges. %
Thus, we cannot improve on \cref{thm:simple-dp} by replacing the 
parameter \emph{number of nodes in the request} with a smaller parameter like \emph{vertex cover number},
\emph{feedback edge number}, or \emph{maximum degree}, unless P$=$NP.

Our algorithm for \cref{thm:simple-dp} works in three steps (see \cref{alg:dynprogram} for a pseudocode illustration):
\begin{enumerate}
	\item Introduce additional leaves to the substrate to ensure that all non-leaves have capacity zero (\cref{lem:map-to-leafs}, method \texttt{Leaf} in the pseudocode).
	\item Split nodes in the substrate with more than two children such that we obtain a binary tree (\cref{lem:bin-tree}, method \texttt{Split}).
	\item Use dynamic programming to solve \minVNEP{} with the substrate being restricted to such trees (method \texttt{GetEntry}).
\end{enumerate}

{
\begin{algorithm}
\DontPrintSemicolon
\SetKwInOut{Input}{Input}\SetKwInOut{Output}{Output}
\SetKwProg{Fn}{Function}{:}{\KwRet}
\SetKwProg{Main}{Main Procedure}{:}{\KwRet}
\SetKwFunction{PostOrderTraversal}{PostOrderTraversal}{}{}
\SetKwFunction{ConstructSolution}{InducedMapping}{}{}
\SetKwFunction{LocallyValid}{InducedMappingLocallyValid}{}{}
\SetKwFunction{PrecomputeShortestValidPaths}{PrecomputeShortestValidPaths}{}{}
\SetKwFunction{MappingCost}{InducedCost}{}{}
\SetKwFunction{Split}{Split}{}{}
\SetKwFunction{ComputeEntry}{GetEntry}{}{}
\SetKwFunction{MakeLeaf}{Leaf}{}{}
\Fn(\tcp*[f]{see \cref{lem:map-to-leafs}}){\MakeLeaf{$v \in \SV$}}{
	Add node~$v'$ to~$\SG$ as a child of~$v$.\\
	$\Scap(v') \gets \Scap(v)$, $\Scost(v') \gets \Scost(v)$.\\
	$\Scap(v) \gets 0$, $\Scost(v) \gets \infty$.\\
	$\Scap(v, v'), \Scap(v', v) \gets \infty$, $\Scost(v, v'), \Scost(v', v) \gets 0$.\\
}
\Fn(\tcp*[f]{see \cref{lem:bin-tree}}){\Split{$v \in \SV$}}{
	Let~$u_1, \dots, u_t$ be the children of~$v$, let~$s = \lfloor t/2 \rfloor$.\\
	Add nodes~$v_\ell$, $v_r$ to~$\SG$, with~$\Scap(v_\ell),\Scap(v_r) \gets 0$ and~$\Scost(v_\ell),\Scost(v_r) \gets \infty$.\\
	Make~$v_\ell$ parent of~$u_1,\!\makebox[1em][c]{.\hfil.\hfil.}, u_s$ (keep capacities and costs).\\
	Make~$v_r$ parent of~$u_{s+1},\!\makebox[1em][c]{.\hfil.\hfil.}, u_t$ (keep capacities and costs).\\
	Make~$v$ parent of~$v_\ell$, $v_r$.\\
	$\Scap(v, v_r), \Scap(v_r, v), \Scap(v, v_\ell), \Scap(v_\ell, v) \gets \infty$.\\
	$\Scost(v, v_r), \Scost(v_r, v), \Scost(v, v_\ell), \Scost(v_\ell, v) \gets 0$.\\
	\lIf{$v_\ell$ has more than 2 children}{call \Split{$v_\ell$}}
	\lIf{$v_r$ has more than 2 children}{call \Split{$v_r$}}
}
\Fn{\ComputeEntry{$R \subseteq \VV, v \in \SV$}}{
	\tcp{returns the of entry in~$D$, or computes it}
	\lIf{$D[R, v]$ was already computed}{\Return $D[R, v]$.}
	\uIf{$v$ is a leaf}{
		$D[R, v] \gets \begin{cases}
			\infty, \;\;\;\;\;\;\;\;\;\; \text{if } \sum_{u \in R} \Vcap(u) \not \leq \Scap(v),\\
			\sum_{u \in R} \Vcap(u)^\top\Scost(v),  \;\;\;\text{ otherwise.}
		\end{cases}$
	}
	\lElseIf{$v$ has one child~$u$}{
		$D[R, v] \gets f(v, u, R)$.
	}
	\ElseIf{$v$ has two children~$u$ and~$w$}{
		$D[R,v] \gets \min\limits_{A \uplus B = R} f(v,w,A)+f(v,u,B)$.\\
		\tcp{Use~$f$ as defined in \eqref{eq:dp-f}, but replace~$D[R, x]$ with \ComputeEntry{$R, x$}.}
	}
	\Return $D[R, v]$ (and mark it as computed).\\
}
\Main{($\SG, \VG, \Vcap, \Scap, \Scost$)}{
	Let~$\SG$ be rooted at some node~$p$.\\
	\For{$v \in V(\SG)$}{
		\lIf{$v$ is not a leaf and $\Scap(v) > 0$}{call \MakeLeaf{$v$}}
	}
	\For{$v \in V(\SG)$}{
		\lIf{$v$ has more than two children}{call \Split{$v$}}
	}
	Initialize table~$D[R, v]$ for all~$R \subseteq \VV$ and~$v \in \SV$.\\
	\Return \ComputeEntry{$\VV, p$}.
}
\nonl
\caption{Algorithm for $\VNEP$ on tree substrates}\label{alg:dynprogram}
\hspace{-12pt}
\end{algorithm}
}

We remark that the first two steps (\cref{lem:map-to-leafs,lem:bin-tree}) can be used as a preprocessing for any algorithms that only work for binary tree substrates on which the capacity of all non-leaf nodes is zero~\cite{oktopus,ccr15emb} to make them work for general tree substrates.

Throughout this section we assume without loss of generality that our substrate graph is bidirectional, that is, for every edge~$(u, v)$ in $\SE$ we also have the edge~$(v, u)$.
Otherwise, we add the missing edge and set its capacity to zero.
Further, we assume that our substrate graph~$\SG$ is a tree rooted at some vertex~$p$.

\paragraph{Introducing additional leaves.}
We first show that we can assume that all non-leaf nodes of our substrate have capacity zero.

\begin{lemma}%
	\label{lem:map-to-leafs}
	Given an instance $I=(\SG,\VG,\Vcap,\Scap,\Scost)$ of \minVNEP{},
	we can build in linear time an instance $\widetilde I=(\tildeSG,\VG,\Vcap,\tildeScap,\tildeScost)$ of \minVNEP{}
	such that
	\begin{enumerate}[label=(\roman*)]
		\item each node $v \in \tildeSV$ of degree at least two fulfills $\tildeScap(v) = 0$, and
		\item there is a solution for $I$ of cost at most $k$
			if and only if
			there is a solution for $\widetilde I$ of cost at most $k$.
	\end{enumerate}
\end{lemma}
{
\begin{proof}
The idea is to add a fresh leaf for each non-leaf vertex with capacities above zero.
Without loss of generality, we assume that each edge in $G_S$ is bidirectional, 
otherwise we add the missing edge to which nothing can be mapped.
We assume that $G_S$ is rooted at some arbitrary node to avoid ambiguity in the following construction 
about whether a neighbor is a child or the parent.
We construct $\tildeSG$ from $G_S$ by
adding a node $v'$ and edges $(v,v'),(v',v)$ for each node $v \in \SV$ which has children
and set
$\tildeScap(v) := 0$,
$\tildeScap(v') := \Scap(v)$,
$\tildeScost(v') := \Scost(v)$,
$\tildeScap(v,v') := \tildeScap(v',v) := \infty$, and
$\tildeScost(\{v,v'\}) := \tildeScost(v',v) := 0$.
Note that we add at most $O(|\SV|)$ nodes and edges to $G_S$. 
Hence, $\widetilde I$ can be constructed after linear time.
We now show that $I$ has a feasible mapping~$(\mapV, \mapE)$ of cost at most $k$ 
if and only if 
$\widetilde I$ has a feasible mapping~$(\tildemapV, \tildemapE)$ of cost at most $k$.

\smallskip
($\Rightarrow$):
Let $(\mapV,\mapE)$ be a solution for $I$ of cost at most $k$.
For all $v \in \VV$, we set $\tildemapV(v) := \mapV(v)$ if $\mapV(v)$ is of degree at most one,
otherwise we set $\tildemapV(v)$ to the new leaf $\mapV(v)'$ of $\mapV(v)$.
For all~$(u, v) \in \VE$, we set~$\tildemapE(u, v)$ to be the unique path from~$\tildemapV(u)$ to~$\tildemapV(v)$ in~$\tildeSG$.
Note that $(\tildemapV,\tildemapE)$ is a solution for~$\widetilde I$ which has the same cost as $(\mapV,\mapE)$.

\smallskip
($\Leftarrow$):
Let $(\tildemapV,\tildemapE)$ be a solution for $\widetilde I$ of cost at most $k$.
For all $v \in \VV$, we set $\mapV(v) := \tildemapV(v)$ if $\tildemapV(v) \in \SV$,
otherwise $\tildemapV(v)$ is a leaf in $\tildeSG$ and we set $\mapV(v)$ to be the parent of $\tildemapV(v)$.
For all~$(u, v) \in \VE$ we set~$\mapE(u, v)$ to be the unique path from~$u$ to~$v$ in~$\SG$.
Observe that the paths induced by~$\mapE(u, v)$ and by~$\tildemapE(u, v)$ may only differ in the leaves that were (possibly) added to the endpoints.
Thus, by construction, $(\mapV, \mapE)$ is a solution for~$I$ of cost at most~$k$.
\end{proof}
}

\paragraph{Splitting non-leaf nodes.}
Next, we show how to turn the substrate into a binary tree.

\begin{lemma}
\label{lem:bin-tree}
Given an instance $I=(G_S,G_R,\Vcap,\Scap,\Scost)$ of \minVNEP{} with~$\SG$ being a tree,
we can construct in linear time an instance $\widetilde I=(\tildeSG,G_R,\Vcap,\tildeScap,\tildeScost)$ of \minVNEP{} 
such that $\tildeSG$ is a binary tree and
there is a solution for $I$ of cost at most $k$ 
if and only if
there is a solution for $\widetilde I$ of cost at most $k$.
\end{lemma}
\begin{proof}
In a nutshell, we are going to replace a node with more than two children with a binary tree of sufficient size.

To construct $\tildeSG = (\tildeSV, \tildeSE)$ from $\SG$, as long as there is a node~$v$ with~$c>2$ children, we replace it with a fresh rooted bidirectional binary tree~$T_v$ with root~$v'$ and $c$ leaves.
We add an edge between $v'$ and the parent of~$v$,
and we add an edge between each child of~$v$ and one designated leaf of~$T_v$.
Furthermore, we set the capacity and cost of the root of $T_v$ to $\Scap(v)$ and $\Scost(v)$, respectively.
All other nodes of $T_v$ get capacity zero and cost $k+1$.
All edges of $T_v$ get capacity $\infty$ and cost zero.
Let~$v_p$ be the parent of~$v$, let~$u$ be a child of~$v$, and let~$v_u$ be the leaf node in~$T_v$ which is adjacent to~$u$.
Then, we set 
$\tildeScap(v',v_p)  := \Scap(v,v_p)$,
$\tildeScap(v_p,v')  := \Scap(v_p,v)$,
$\tildeScost(v',v_p) := \Scost(v,v_p)$,
$\tildeScost(v_p,v') := \Scost(v_p,v)$,
$\tildeScap(v_u,u)   := \Scap(v,u)$,
$\tildeScap(u,v_u)   := \Scap(u,v)$,
$\tildeScost(v_u,u)  := \Scost(v,u)$, and
$\tildeScost(u,v_u)  := \Scost(u,v)$.
All other values of $\tildeScap$ and $\tildeScost$ are equal to $\Scap$ and $\Scost$, respectively. 

Note that by the handshake lemma (the sum of degrees of is twice the number of edges in a graph), $\tildeSG$ is of size $O(|\SV|)$, because $T_v$ is of size $O(\deg_{G_S}(v))$.
Hence, we can construct $\widetilde I$ in linear time.

We show that $I$ has a feasible mapping~$(\mapV, \mapE)$ of cost at most $k$ 
if and only if 
$\widetilde I$ has a feasible mapping~$(\tildemapV, \tildemapE)$ of same cost.%

($\Rightarrow$):
Let $(\mapV,\mapE)$ be a feasible mapping for $I$ of cost at most $k$.
For all $v \in \VV$, we set $\tildemapV(v) := \mapV(v)$ if $\mapV(v) \in \tildeSV$,
otherwise we set $\mapV(v)$ to be the root of $T_v$.
Hence, we have for all $v \in \tildeSV$ that $\sum_{w : \tildemapV(w) = v} \Vcap(w) \leq \tildeScap(v)$.
For all~$(u, v) \in \VE$, we set~$\tildemapE(u, v)$ to be the unique path from~$\tildemapV(u)$ to~$\tildemapV(v)$ in~$\tildeSG$ (recall that~$\tildeSG$ is a tree).
So for all $e_S \in \tildeSE$ we have $\sum_{e_R:e_S \in \tildemapE(e_R)} \Vcap(e_R) \leq \tildeScap(e_S)$ 
and for all $(u,v) \in \VE$ we have that $\tildemapE(u,v)$ is a path from~$\tildemapV(u)$ to~$\tildemapV(v)$.
Moreover, by our construction, we get that 
\begin{align*}
   & \sum_{v \in \VV} \Vcap(v)^\top \tildeScost(\tildemapV(v)) + 
   \sum_{e \in \VE} \Big(\sum_{e' \in \tildemapE(e)} \Vcap(e)^\top\tildeScost(e') \Big) \\   
		= &
\sum_{v \in \VV} \Vcap(v)^\top \Scost(\mapV(v)) + 
\sum_{e \in \VE} \Big(\sum_{e' \in \mapE(e)} \Vcap(e)^\top\Scost(e') \Big)
\leq k.
\end{align*}
Thus, $(\tildemapV,\tildemapE)$ is a feasible mapping for $\widetilde I$ of cost at most $k$.

\smallskip
($\Leftarrow$):
Let $(\tildemapV,\tildemapE)$ be a feasible mapping for $\widetilde I$ of cost at most $k$.
Let $v \in \VV$.
Note that if $\tildemapV(v)\notin \SV$, then there must be a node $w \in \SV$ 
such that $\tildemapV(v)$ is a node in~$T_w$.
Hence, we set~$\mapV(v) := \tildemapV(v)$ if $\tildemapV(v)\in \SV$,
otherwise we set~$\mapV(v) := w$, where~$w \in \SV$ is the node replaced by~$T_w$ and~$\tildemapV(v)$ is a node of~$T_w$. 
So, for all~$v \in \SV$, we have $\sum_{w:\mapV(w)=v} \Vcap(w) \leq \Scap(v)$.
For all~$(u, v) \in \VE$ we set~$\mapE(u, v)$ to be the unique path in~$\SG$ from~$\mapV(u)$ to~$\mapV(v)$.
Note that the path induced by~$\tildemapE(u, v)$ consists of a subset of edges of~$\mapE(u, v)$; thus for all~$e_S \in \SE$ we have $\sum_{e_R:e_S \in \mapE(e_R)} \Vcap(e_R) \le \Scap(e_S)$.
Moreover, we have that
\begin{align*}
 & \sum_{v \in \VV} \Vcap(v)^\top \Scost(\mapV(v)) + \sum_{e \in \VE} \Big(\sum_{e' \in \mapE(e)} \Vcap(e)^\top\Scost(e') \Big) \\
		= &
    \sum_{v \in \VV} \Vcap(v)^\top\tildeScost(\tildemapV(v)) + \sum_{e \in \VE} \Big(\sum_{e' \in E(\tildemapE(e))} \Vcap(e)^\top \tildeScost(e') \Big)
\leq k.
\end{align*}
Thus, $(\mapV,\mapE)$ is a feasible mapping for $I$ of cost at most $k$.
\end{proof}

\paragraph{Dynamic program.}
Now that we have created an instance in which the substrate is a binary tree in which only the leaf nodes have nonzero capacity, we can formulate our dynamic program.
Let~$p$ be the root of~$\SG$.
For each $v \in \SV$, let $T_v$ be the induced subtree of $G_S$ where $v$ is the root, 
that is, $T_v$ contains all nodes $u$ for which the path from~$u$ to~$p$ visits~$v$.
We assume that $G_S$ is a full binary tree, i.e., each node is either a leaf or has two children
(otherwise we add a fresh leaf to which nothing can be mapped).

Removing the edges $(v,u),(u,v) \in \SE$ splits the tree $G_S$ into two rooted trees.
Without loss of generality assume that $v$ is the parent of $u$ in $G_S$.
Hence, one of the trees is $T_u$ and the other one is $T' := \SG[\SV \setminus V(T_u)]$.
Note that for a given solution $(\mapV,\mapE)$ of $I$,
the cut $\{(v,u),(u,v)\}$ also splits the mapping of $G_R$ into two parts $B := \{ w \in \VV \mid \mapV(w) \in V(T_u) \}$ and 
$A := \VV \setminus B$.
Further, for each edge $e \in \VE$ we have that $(v,u) \in \mapE(e)$ 
if and only if $e \in \cut_{G_R}(A) := \{ (x,y) \in \VE \mid x \in A, y \not \in A \}$,
and moreover $(u,v) \in \mapE(e)$
if and only if $e \in \cut^-_{G_R}(A) := \cut_{G_R}(\VV \setminus A)$,
since every path from $T'$ to $T_u$ must contain $(v,u)$ and
every path from $T_u$ to $T'$ must contain $(u,v)$.
We use this observation to describe a dynamic program 
in which each entry $D[R,v]$ contains the minimum cost for a feasible mapping of $G_R[R]$ into $T_v$ plus the induced cost of $\cut_{G_R}(A) \cup \cut^-_{G_R}(A)$ on edges in $T_v$.

Let $v \in \SV$ and $R \subseteq \VV$.
If $v$ is a leaf, then
\begin{align}
		\label{eq:dp-leaf}
 	D[R,v] \coloneqq
	\begin{cases}
		\infty, & \text{if } \sum_{u \in R} \Vcap(u) \not \leq \Scap(v)\\
		\sum_{u \in R} \Vcap(u)^\top\Scost(v),  & \text{ otherwise.}
	\end{cases}
\end{align}
If $v$ is not a leaf, then
\begin{align}
		\label{eq:dp-non-leaf}
		D[R,v] \coloneqq \min_{A \uplus B = R} f(v,w,A)+f(v,u,B),
\end{align}
where $w$ and $u$ are the neighbors of $v$ in $T_v$ and for $x \in \{w,u\}$ the function $f$ is defined as \ifarxiv{}\fi{}
\ifarxiv{}
\begin{align}
		\label{eq:dp-f}
		f(v,x,R)\coloneqq 
		\begin{cases}
				\infty,  \,\ \;\;\;\;\;\qquad\qquad\qquad\text{ if }\displaystyle \sum_{ e \in \cut_{G_R}(R)} \Vcap(e) \not\leq \Scap(x,v),\\
				\infty,  \,\ \;\;\;\;\qquad\qquad\qquad\text{ if }\displaystyle \sum_{ e \in \cut^-_{G_R}(R)} \Vcap(e) \not\leq \Scap(v,x),\\
				D[R,x] + %
				\displaystyle \sum_{ e \in \cut^-_{G_R}(R)\cup \cut_{G_R}(R)} \!\!\!\!\!\!\!\!\!\Vcap(e)^\top\Scost(v,x)				%
				, \hfill \text{ otherwise.}
		\end{cases}	
\end{align}
\fi{}
To show the correctness of the dynamic program (defined in \eqref{eq:dp-leaf} and \eqref{eq:dp-non-leaf}), 
we introduce the following notations and definitions.
For~$v \in \SV$, for~$(x,y) \in E(T_v)$, for~$X \subseteq \VV$, and $\mapV \colon X \rightarrow V(T_v)$,
let~$\mathcal P_{(x,y)}^v(X)$ be the set of paths $P$ within~$T_v$ between~$v$ and a node~$\mapV(u^*)$ 
such that $(x,y)$ is in $P$, and if $v$ is the start node of $P$, $u^* \in X$ is the sink of an edge in~$\cut^-_{G_R}(X)$, 
otherwise $u^* \in X$ is the source of an edge in $\cut_{G_R}(X)$.
Furthermore, let \ifarxiv{}\fi{}
\ifarxiv{}
\begin{align*}
	E_{(x,y)}^v(X) := & \{ (u^*,w^*) \in \cut_{G_R}(X) \mid (x,y) \text{ is on the $\mapV(u^*)$--$v$-path in $T_v$} \}\\
	\cup	& \{ (w^*,u^*) \in \cut^-_{G_R}(X) \mid (x,y) \text{ is on the $v$--$\mapV(u^*)$-path in $T_v$} \}. 
\end{align*}
\fi{}
\begin{definition}
	For a node $v \in \SV$ and a subset $X \subseteq \VV$.
	We call a feasible mapping $(\mapV,\mapE)$ of $G_R[X]$ to $T_v$ \emph{desirable} if
		for every edge~$e_S \in E(T_v)$ we have
		\begin{equation}
			\label{eq:dp-desirable-1}
			\sum_{e_R : e_S \in \mapE(e_R)} \Vcap(e_R) \le \Scap(e_S) - \sum_{e \in E_{e_S}^v(X)} d_R(e).
		\end{equation}
	Furthermore, we say that the \emph{induced cost} of $(\mapV,\mapE)$ is
		\begin{equation}
			\label{eq:dp-desirable-2}
			\begin{aligned}
				&\sum_{w \in X} \Vcap(w)^\top\Scost(\mapV(w)) + 
				\sum_{e \in  E(G_R[X])} \Big(\sum_{e' \in \mapE(e)} \Vcap(e)^\top\Scost(e')\Big)+
				\\
				&\sum_{e \in \cut_{G_R}(X)} \!\Big(\sum_{e' \in P_e} \!\!\Vcap(e)^\top\Scost(e')\Big) 
					+\sum_{e \in \cut^-_{G_R}(X)} \!\Big(\sum_{e' \in P^-_e} \!\!\Vcap(e)^\top\Scost(e')\Big).
			\end{aligned}
		\end{equation}
		Here~$P_e$ is the set of edges of the path from the source of~$e$ to~$v$ in~$T_v$
		and~$P^-_e$ is the set of edges of the path from~$v$ to the target of~$e$ in~$T_v$.
\end{definition}
Later, our algorithm will report that the minimum cost for a solution is $D[\VV,p]$.
We show that indeed there is such a solution.
\begin{lemma}%
		\label{lem:dp-induction1}
		Let $v \in \SV$ and $X \subseteq \VV$.
		If $D[X,v] < \infty$, then there is a desirable mapping $(\mapV,\mapE)$ of
		$G_R[X]$ onto $T_v$ where the induced cost is at most $D[X,v]$.
\end{lemma}
{
\begin{proof}
		We show this by induction over the tree $G_S$.
		By the definition in \eqref{eq:dp-leaf}, every mapping of~$G_R[X]$, $X \subseteq \VV$, onto a leaf~$v \in \SV$ is desirable
		and has induced costs of $D[X,v]$.

		For the induction step, let $v \in \SV$ be a non-leaf, let $X \subseteq \VV$,
		and assume that for all~$u \in V(T_v) \setminus \{v\}$ we have that if~$D[Y, u] < \infty$.
		Then there is a desirable mapping of~$G_R[Y]$ onto~$T_u$ with induced cost of at most $D[Y, u]$.
		Assume further that~$D[X, v] < \infty$, and let~$a$ and $b$ be the children of~$v$.
		Then by the definition in \eqref{eq:dp-non-leaf} there is a partition $A \uplus B = X$
		such that\ifarxiv{}\fi{}
		\ifarxiv{}
		\begin{equation}
			\label{eq:ind1}
			\begin{split}
			D[X,v] = D[A,a] + D[B,b] +
				\sum_{ e \in \cut^-_{G_R}(A)} \Vcap(e)^\top\Scost(v,a)+
				   \sum_{ e \in \cut_{G_R}(A)} \Vcap(e)^\top\Scost(a,v)\\
				+ \sum_{ e \in \cut^-_{G_R}(B)} \Vcap(e)^\top\Scost(v,b)
				   + \sum_{ e \in \cut_{G_R}(B)} \Vcap(e)^\top\Scost(b,v).
	\end{split}
		\end{equation}
		\fi{}
		Thus, $D[A,a] < \infty$ and $D[B,b] < \infty$, and we get by assumption that there are desirable mappings $({\mapV}^a,{\mapE}^a)$ and $({\mapV}^b,{\mapE}^b)$ of $G_R[A]$ onto $T_a$
		and of $G_R[B]$ onto $T_b$, respectively. 

		We create a mapping~$(\mapV, \mapE)$ of~$G_R[X]$ onto~$T_v$ with
		\begin{equation}
			\label{eq:dp-mapX}
		\begin{aligned}
			\mapV(x) &:=
				\begin{cases}
					{\mapV}^a(x),& x \in A,\\
					{\mapV}^b(x),& x \in B,
				\end{cases}
			\qquad \qquad \qquad \text{and}\\
			\mapE(x,y) &:=
				\begin{cases}
					{\mapE}^a(x,y),& x,y \in A,\\
					{\mapE}^b(x,y),& x,y \in B,\\
					\text{path from ${\mapV}^a(x)$ to ${\mapV}^b(y)$ in $T_v$}, & x \in A, y \in B,\\
					\text{path from ${\mapV}^b(x)$ to ${\mapV}^a(y)$ in $T_v$}, & x \in B, y \in A.
				\end{cases}
		\end{aligned}
		\end{equation}

		Observe that $(\mapV, \mapE)$ is a feasible mapping of~$G_R[X]$ onto~$T_v$:
		Let~$(x, y)$ be an edge in~$E(G_R[X])$ such that one endpoint is in~$A$ and the other endpoint is in~$B$.
		Then every edge in~$T_v$ that is on a path from~$\mapV(x)$ to~$\mapV(y)$ has sufficient capacity to map all edges of~$\mapE(x, y)$ as defined in \eqref{eq:dp-mapX}.
		Hence, \eqref{eq:dp-desirable-1} for $({\mapV}^a,{\mapE}^a)$ and $({\mapV}^b,{\mapE}^b)$
		implies that 
		for every edge~$e_S \in E(T_v)$ we have
		\begin{equation*}
		\ifarxiv{}
			\sum_{e_R : e_S \in \mapE(e_R)} \Vcap(e_R) \le \Scap(e_S) - \sum_{e \in E_{e_S}^v(X)} d_R(e).
		\fi{}
		\end{equation*}
 
		Moreover, for all $c \in \{a,b\}$, a path from a node in $V(T_c)$ to $v$ contains the edge $(v,c)$ and
		a path from $v$ to some node in $V(T_c)$ contains the edge $(c,v)$.
		Hence, the induced cost of $(\mapV,\mapE)$ is the sum of
		the induced cost of $({\mapV}^a,{\mapE}^a)$ and $({\mapV}^b,{\mapE}^b)$ and
		\begin{align*}
			&
			\sum_{ e \in \cut^-_{G_R}(A)} \Vcap(e)^\top\Scost(v,a)
			+ \sum_{ e \in \cut_{G_R}(A)} \Vcap(e)^\top\Scost(a,v) + 
			\\
			&
			\sum_{ e \in \cut^-_{G_R}(B)} \Vcap(e)^\top\Scost(v,b)
			+ \sum_{ e \in \cut_{G_R}(B)} \Vcap(e)^\top\Scost(b,v).
		\end{align*}
		Thus, by \eqref{eq:ind1} the induced cost of $(\mapV,\mapE)$ is at most $D[X,v]$, 
		because the induced cost of $({\mapV}^a,{\mapE}^a)$ is at most $D[A,a]$ and
		the induced cost of $({\mapV}^b,{\mapE}^b)$ is at most $D[B,b]$.
		
		Finally, since $D[X,v] < \infty$ we get by \eqref{eq:dp-non-leaf}
		that \eqref{eq:dp-desirable-1} holds for $(\mapV,\mapE)$ as well.
		Thus, $(\mapV,\mapE)$ is a desirable mapping of~$G_R[X]$ onto~$T_v$ of induced cost at most $D[X,v]$, and we are done.
\end{proof}
}
\allowdisplaybreaks
Moreover, we also need to show that 
if there is feasible mapping for $I$ of cost at most $k$, then $D[\VV,p] \leq k$.
More formally, we show:

\begin{lemma}%
		\label{lem:dp-induction2}
		Let $v \in \SV$ and $(\mapV,\mapE)$ be a feasible mapping for $I$ of cost at most $k$.
		Then,
		\ifarxiv{}
		\begin{align*}
			D[X, v] &\le \sum_{w \in X} \Vcap(w)^\top\Scost(\mapV(w))\\
				&+
						\sum_{e \in \cut_{G_R}(X) \cup \cut^-_{G_R}(X)  \cup E(G_R[X])} 
			\Big(
				\sum_{e' \in \mapE(e) \cap E(T_v)}
				\Vcap(e)^\top 
				\Scost(e')
			\Big),
		\end{align*}
		\fi{}
		where $X := \{ w \in \VV \mid \mapV(w)\in V(T_v)\}$.
\end{lemma}
{
\begin{proof}
		We show the statement of the lemma by structural induction over the tree $G_S$.
		By the definition  in \eqref{eq:dp-leaf}, this is true for all leaves $v \in G_S$ as $E(T_v) = \emptyset$.

		For the induction step let $v \in \SV$ be a non-leaf node, let $X := \{ w \in \VV \mid \mapV(w)\in V(T_v)\}$,
		and assume that 
		for all nodes $u \in V(T_v) \setminus \{v\}$ we have 
		\ifarxiv{}
		\begin{equation}
				\begin{split}
			\label{eq:dp-assumption-Y}
			D[Y, u] \le
				\sum_{w \in Y}\Vcap(w)^\top\Scost(\mapV(w))
				+
				\sum_{e \in \cut_{G_R}(Y)\cup \cut^-_{G_R}(Y)\cup E(G_R[Y])} 
				\Big(
		\sum_{e' \in \mapE(e) \cap E(T_u)} \hspace{-12pt}
		\Vcap(e)^\top\Scost(e')\Big),
				\end{split}
		\end{equation}
		\fi{}
		where $Y:=\{ w \in \VV \mid \mapV(w)\in V(T_u)\}$.
		Now let $a$ and $b$  be the children of $v$, and 
		let $A := \{ w \in \VV \mid \mapV(w)\in V(T_a)\}$ and 
		$B :=\{ w \in \VV \mid \mapV(w)\in V(T_b)\}$.
		Node~$v$ is not a leaf; thus~$\Scap(v)=0$, that is, no node of~$G_R$ can be mapped onto~$v$.
		By the definition in \eqref{eq:dp-non-leaf} we obtain \ifarxiv{}\fi{}
		\ifarxiv{}
		\begin{align*}
				D[X,v] &\leq 
				D[A,a] + \sum_{e \in \cut^-_{G_R}(A)} \Vcap(e)^\top\Scost(v,a)
				+ \sum_{e \in \cut_{G_R}(A)} \Vcap(e)^\top\Scost(a,v)
					   +\\
				&+D[B,b] + \sum_{e \in \cut^-_{G_R}(B)} \Vcap(e)^\top\Scost(v,b)
				+ \sum_{e \in \cut_{G_R}(B)} \Vcap(e)^\top\Scost(b,v).
		\end{align*}
		\fi{}
		By assumption, \eqref{eq:dp-assumption-Y} holds for~$D[A,a]$ and~$D[B,b]$; so
		\ifarxiv{}
		\begin{align*}
				D[X,v] &\leq  \sum_{w \in A \cup B} \Vcap(w)^\top\Scost(\mapV(w))\\
				&+ \sum_{e \in \cut_{G_R}(A)\cup \cut^-_{G_R}(A) \cup E(G_R[A])} \Big(\sum_{e' \in \mapE(e) \cap E(T_a)} \Vcap(e)^\top\Scost(e')\Big)\\
				&+ \sum_{e \in \cut_{G_R}(B)\cup\cut^-_{G_R}(B) \cup E(G_R[B])} \Big( \sum_{e' \in \mapE(e) \cap E(T_b)}\Vcap(e)^\top\Scost(e')\Big)\\
				&+\sum_{e \in \cut^-_{G_R}(A)} \Vcap(e)^\top\Scost(v,a)
							+ \sum_{e \in \cut_{G_R}(A)} \Vcap(e)^\top\Scost(a,v)\\
				&+ \sum_{e \in \cut^-_{G_R}(B)} \Vcap(e)^\top\Scost(v,b)+
							\sum_{e \in \cut_{G_R}(B)} \Vcap(e)^\top\Scost(b,v).
		\end{align*}
		\fi{}
		Note that every path from a node in~$T_a$ ($T_b$) to a node in~$T_b$ ($T_a$) 
		contains the edges $(a,v),(v,b)$ ($(b,v),(v,a)$).
		Moreover, for $c \in \{a,b\}$ every path from $T_c$ to some node in $G_S - V(T_v)$ contains the edge $(c,v)$
and every path from $G_S - V(T_v)$ to some node in $T_c$ contains the edge $(v,c)$.
		Hence, we obtain 
		\ifarxiv{}
		\begin{align*}
				D[X,v] &\leq \sum_{w \in A \cup B} \Vcap(w)^\top\Scost(\mapV(w))+ \sum_{e \in  E(G_R[X])} \Big(\sum_{e' \in \mapE(e) \cap E(T_v)} \Vcap(e)^\top\Scost(e')\Big)\\
				&+\sum_{e \in \cut_{G_R}(X) \cup \cut^-_{G_R}(X)} 
			\Big(
				\sum_{e' \in \mapE(e) \cap E(T_v)}
				\Vcap(e)^\top 
				\Scost(e')
			\Big). \qedhere
		\end{align*}
		\fi{}
\end{proof}
}

\begin{figure*}[t]
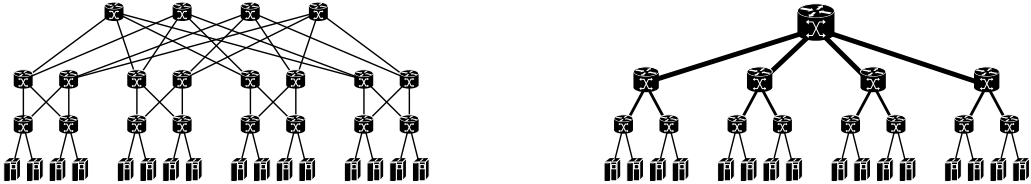

		\centering
\ifarxiv{}
	\begin{subfigure}[t]{0.45\textwidth}
		\centering
		\includegraphics[width=0.8\textwidth]{fattree_actual.pdf}
	\end{subfigure}
	\qquad
	\begin{subfigure}[t]{0.45\textwidth}
		\centering
		\includegraphics[width=0.8\textwidth]{fattree_abstraction.pdf}
	\end{subfigure}
\fi{}
\vspace{-6pt}
	\caption{Fat tree topology~\cite{fattrees} constructed using $f=4$-port switches (left) and corresponding forwarding abstraction (right).}
	\label{fig:fattree}
	\vspace{-6pt}
\end{figure*}

Now we have everything at hand to prove \cref{thm:simple-dp}.

\begin{proof}[Proof of \cref{thm:simple-dp}]
	Let $I=(G_S,G_R,\Vcap,\Scap,\Scost)$ be some instance of \minVNEP{}.
	By \cref{lem:bin-tree,lem:map-to-leafs} we can assume that $G_S$ is a binary tree 
	rooted at some arbitrary node $p$ and  each node $v \in \SV$ with degree at least two fulfills $\Scap(v) = 0$.
	We apply the dynamic program stated in \eqref{eq:dp-leaf} and \eqref{eq:dp-non-leaf}.
	Since  $G_S = T_p$ and $\cut_{G_R}(\VV) = \emptyset$,
	\cref{lem:dp-induction1,lem:dp-induction2} imply that $D[\VV,p]$ contains the minimum cost for a feasible mapping for $I$,
	where $D[\VV,p] = \infty$ if and only if there is no feasible mapping for $I$.

	Let $r := |\VV|$.
	It remains to be shown that $D[\VV,p]$ can be computed in $\mathcal O(3^r(|\SV| + r^2))$ time.
	We first compute for every~$A \subseteq \VV$ the demand of the cut~$\cut_{G_R}(A)$.
	There are~$2^r$ subsets~$A$, for each of which we need to iterate over the~$\mathcal O(r^2)$ edges; thus this step takes~$\mathcal O(2^r \cdot r^2)$ time.
	With this at hand we can compute~$D[X, v]$ in constant time for each leaf~$v \in \SV$ and for each subset~$X \subseteq \VV$.
	For a non-leaf node~$v$, computing the entries~$D[X, v]$ for each~$X \subseteq \VV$ can be done in~$\mathcal O(3^r)$ operations:
	For a partition~$X = A \uplus B$ we require constant time.
	Observe that there are~$3^r$ partitions of~$\VV$ into three parts~$A \uplus B \uplus C$.
	Thus, choosing~$X = \VV \setminus C$ gives us all partitions of all subsets~$X \subseteq \VV$ into two parts~$A$ and~$B$.
	Thus, for all non-leaf nodes~$v$ and all subsets~$X \subseteq \VV$ combined we require~$\mathcal O(3^r \cdot |\SV|)$ time.
	Altogether, this yields the claimed running time of $\mathcal O(3^r \cdot (|\SV| + r^2))$.
\end{proof}

As a final note, we highlight that our dynamic program is rather simple to implement 
and robust in the sense that it also works if one has further natural constraints or other objectives.

\section{Evaluation}
\newcommand{\ErdosRenyi}{Erd\H{o}s-R\'enyi}
\label{sec:eval}
\ifarxiv{}\fi{}

We evaluate the performance of our exact dynamic programming algorithm for tree substrates 
(presented in Section~\ref{sec:robust-dp} and henceforth abbreviated with DP)
on common fat tree topologies as they are widely deployed, e.g., in data centers~\cite{fattrees}.
Specifically, we compare the performance of our algorithm with two well-established approaches for solving the~$\VNEP$.
The first is the standard integer programming formulation (IP) which gives exact results.
The second is the ViNE heuristic by \citet{vnep}, which takes the relaxation of an IP formulation and then applies randomized rounding to fix node mappings and realizes edges via shortest paths.
In our comparisons the focus is on the running time and the solution quality of the three approaches.
Since the running time of the IP may take hours for medium-sized instances, we set a time limit on the IP, and we also report on the quality of the sub-optimal solutions obtained by the IP when the imposed time limit was reached.
Recall that the solution obtained by our DP is \emph{always} optimal.

\paragraph{Testing Methodology.}
\label{sec:eval:methodology}

For our evaluation, we employ fat trees~\cite{fattrees} as our substrate network topology.
Fat trees are common topologies, e.g., in data centers 
 built using commodity switches, %
where each switch has the same number $f \geq 4$ of ports.
Fat trees are highly structured: servers are located at the bottom and are connected by a three-layer hierarchy of switches (see \cref{fig:fattree}).
A fat tree constructed of $f$-port switches connects up to $f^3/4$ servers.
While the actual physical infrastructure is not a tree, \emph{the forwarding abstraction} provided by fat trees is a tree. %
Specifically, based on link aggregation techniques~\cite{link-aggregation}, switches and their interconnections are logically aggregated from an application-level perspective.
Hence, embeddings can and must be computed on this tree forwarding abstraction.
Note that \minVNEP{} is clearly NP-hard on such trees (see \cref{obs:np-hard}).

We consider seven different fat tree forwarding abstractions for $f \in \{4,6,\ldots,16\}$, hosting between 16 and 1024 servers and using between 5 and 145 switches.
Considering a single node resource type, we set the computational capacities on servers to 1 and on switches to 0.
For edges of the bottom layer, i.e., connecting to servers we set a bandwidth of one.
Due to the aggregation of edges, the edge bandwidth of the above layers is set accordingly to $f/2$ and $(f/2)^2$.
To simulate heterogeneous usage patterns within the data center, we perturb node and edge capacities by random factors drawn from the interval $[1,10]$ and draw costs from~$[1,10]$.

For generating requests, we follow the standard approach of sampling \ErdosRenyi{}-topo\-logies of various sizes~\cite{vnep, vnep-survey}.
In this model, for a specific number of nodes, edges between pairs of nodes are created probabilistically using a connection probability~$p$.
This approach is attractive, as it does not impose assumptions on the applications modeled by the requests albeit allowing to easily vary the interconnection density.
Again, following the standard evaluation methodology~\cite{vnep, vnep-survey}, node and edge demands are also sampled uniformly at random.
Specifically, node demands are drawn from the interval $[1,5]$.
For edge demands, we proceed as follows.
For each node, we draw the total cumulative outgoing bandwidth from $[1,5]$ and then distribute the bandwidth randomly across the actual edges.
By this construction, the expected total bandwidth (per request size) is independent from the connection probability~$p$. 

For our evaluation we focus on requests of~5 to~12 nodes and consider ten different connection probabilities $p \in \{0.1, 0.2, \ldots, 1.0\}$ (disconnected graphs are discarded and resampled).
For each combination of graph size and connection probability, we sample ten instances. Together with the $7$ different fat tree topologies, our computational study encompasses 5.6k instances.

\paragraph{Computational Setup.}
\label{sec:eval:baseline-and-setup}
\appendixsection{sec:eval:baseline-and-setup}
We first discuss the implementation of our dynamic program (DP), 
the integer programming (IP),
and ViNE.

We have implemented the dynamic program presented in \cref{sec:robust-dp} in C++ using only the standard library.
While implemented for single node and edge resources, our implementation can be easily extended to an arbitrary number of resources.
Our implementation is tweaked to skip computations that involve table entries containing~$\infty$, 
as these cannot lead to %
a feasible solution.
Furthermore we do not store table entries that contain~$\infty$.
To facilitate this, we store the table entries for a node~$v$ as a set-trie, rather than a simple array, to allow for fast subset and superset queries.
During our experiments we discovered that on instances with 12-node requests this tweak resulted in a decrease of~$90\%$ in table size, and a corresponding drop in the running time is to be expected.
The source of our implementation is available online.\footnote{\url{https://git.tu-berlin.de/akt-public/vnep-for-trees}}

Existing exact algorithms for the $\VNEP$ in the literature are essentially all based on integer programming~\cite{vnep-survey}.
Especially one integer programming formulation, based on multi-commodity flows, has been studied extensively~\cite{vnep,infuhr2011introducing,rostSchmidFeldmann2014,rostSchmidVNEP_RR_IFIP_18}.
\ifarxiv{} \fi{}
\toappendix{
	\ifarxiv{}\paragraph{Integer Program for \minVNEP.}\fi{}
		\label{app:sec:eval:baseline-and-setup}
We revisit the integer programming formulation used in our evaluation, introduced below as Integer Program~\ref{alg:ip:formulation}.
The binary variables $y^u_i \in \{0,1\}$ indicate whether the request node~$i \in \VV$ is mapped onto substrate node $u \in \SV$ 
The binary variables $z^{u,v}_{i,j} \in \{0,1\}$ indicate whether the substrate edge $(u,v) \in \SE$ lies on the path used by the request edge $(i,j) \in \VE$. 
By Constraint~\eqref{alg:ip:node-embedding}, all request nodes must be mapped.
Constraint~\eqref{alg:ip:node-forbidding} forbids the mapping onto nodes not providing sufficient capacities. 
Constraint~\eqref{alg:ip:edge-embedding} induces a unit flow for each request~edge~$(i,j) \in \VE$ between the nodes onto which $i$ and $j$ have been mapped, respectively. 
Constraint~\eqref{alg:ip:edge-forbidding} forbids the mapping of request edges onto substrate edges not providing sufficient capacities and Constraints~\eqref{alg:ip:node-capacity} and~\eqref{alg:ip:edge-capacity} safeguard that capacities are not violated. 
The formulation naturally models the $\minVNEP$ objective.

{
	\renewcommand{\figurename}{Integer Program}
	\begin{figure}[h!]
		\SetAlgorithmName{Integer Program}{}{{}}
		\centering
		\scalebox{0.95}{
			\renewcommand{\arraystretch}{1.}	

			\begin{tabular}{\ipMR \ipMRspace \ipMLspace \ipMLmorespace \ipMR}
				\multicolumn{4}{\ipMC}{ \min \left(\begin{array}{rl}
						&  \sum \limits_{i \in \VV, u \in \SV} y^u_i \left(\Vcap(v)^\top \cdot \Scost(u)\right) \\
						+ & \sum \limits_{(i,j) \in \VE, (u,v) \in \SE } z^{u,v}_{i,j}  \left(\Vcap(i,j)^\top\cdot \Scost(u,v)\right)
					\end{array}\right)} &  \tagIt{alg:ip:obj} \\
				
				\sum_{u \in \SV } y^u_{i} & = & 1 & \forall i \in \VV &  \tagIt{alg:ip:node-embedding} \\
				
				\sum_{u \in \SV: \Scap(u) \not \leq \Vcap(i)} y^u_{i} & = & 0 & \forall i \in \VV &  \tagIt{alg:ip:node-forbidding} \\
				
				\left( \begin{array}{l}
					\sum \limits_{(u,v) \in \cut^+_{\SG}(u)}   z^{u,v}_{i,j} \\
					-  \sum \limits_{(v,u) \in \cut^-_{\SG}(u)}  z^{v,u}_{i,j}
				\end{array} \right) & = & y^u_{i} - y^u_{j} & \forall (i,j) \in \VE, u \in \SE &  \tagIt{alg:ip:edge-embedding} \\
				
				\sum_{(u,v) \in \SE: \Scap(u,v) \not \leq \Vcap(i,j) } \hspace{-6pt} z^{u,v}_{i,j} & = & 0 & \forall (i,j) \in \VE &  \tagIt{alg:ip:edge-forbidding} \\
				
				\sum_{i \in \VV} \Vcap(i) \cdot y^u_{i} & \leq & \Scap(u) & \forall u \in \SV & \tagIt{alg:ip:node-capacity} \\
				
				\sum_{(i,j) \in \VE} \Vcap(i,j) \cdot z^{u,v}_{i,j} & \leq & \Scap(u,v) & \forall (u,v) \in \SE & \tagIt{alg:ip:edge-capacity} \\
				
				y^u_i & \in & \{0,1\} & \forall i \in \VV, u \in \SV & \tagIt{alg:ip:var-y}\\
				z^{u,v}_{i,j} & \in & \{0,1\} & \forall (i,j) \in \VE, (u,v) \in \SE & \tagIt{alg:ip:var-z}
			\end{tabular}
		}
		\caption{MCF $\minVNEP$ Formulation}
		\label{alg:ip:formulation}	
	\end{figure}
}

}%

To construct the integer program, we employ a simple GMPL model and translate it into an LP-file using GLPSOL. We then solve the integer program using the commercial solver Gurobi~8.1.1.  We set the thread limit of Gurobi to 1, to allow for a fair comparison with the single-threaded dynamic program.
As the running time of the IP drastically exceeds the running time of the DP, for each instance we employ a time limit of 200 times the running time of the DP. Notably, the time to construct the LP-files using the unoptimized GLPSOL command is not counted towards the running time of the IP, as it often exceeded it by a factor of 3 even on smaller instances.

For the ViNE baseline, we use the Python 2 implementation of \citet{rostSchmidFeldmann2014} with Gurobi 8.1.1 to solve the LP relaxation. 
Given a solution for the LP relaxation, we try~25 times to obtain a feasible solution by randomized rounding.
For more details on the ViNE heuristic, we refer to \citet{vnep} and \citet{rostSchmidFeldmann2014}.

\paragraph{Results.}
\label{sec:eval:results}
\ifarxiv{}
\begin{figure*}[t]
	\centering
	\begin{subfigure}[t]{0.32\textwidth}
			\centering
			\includegraphics[width=\textwidth]{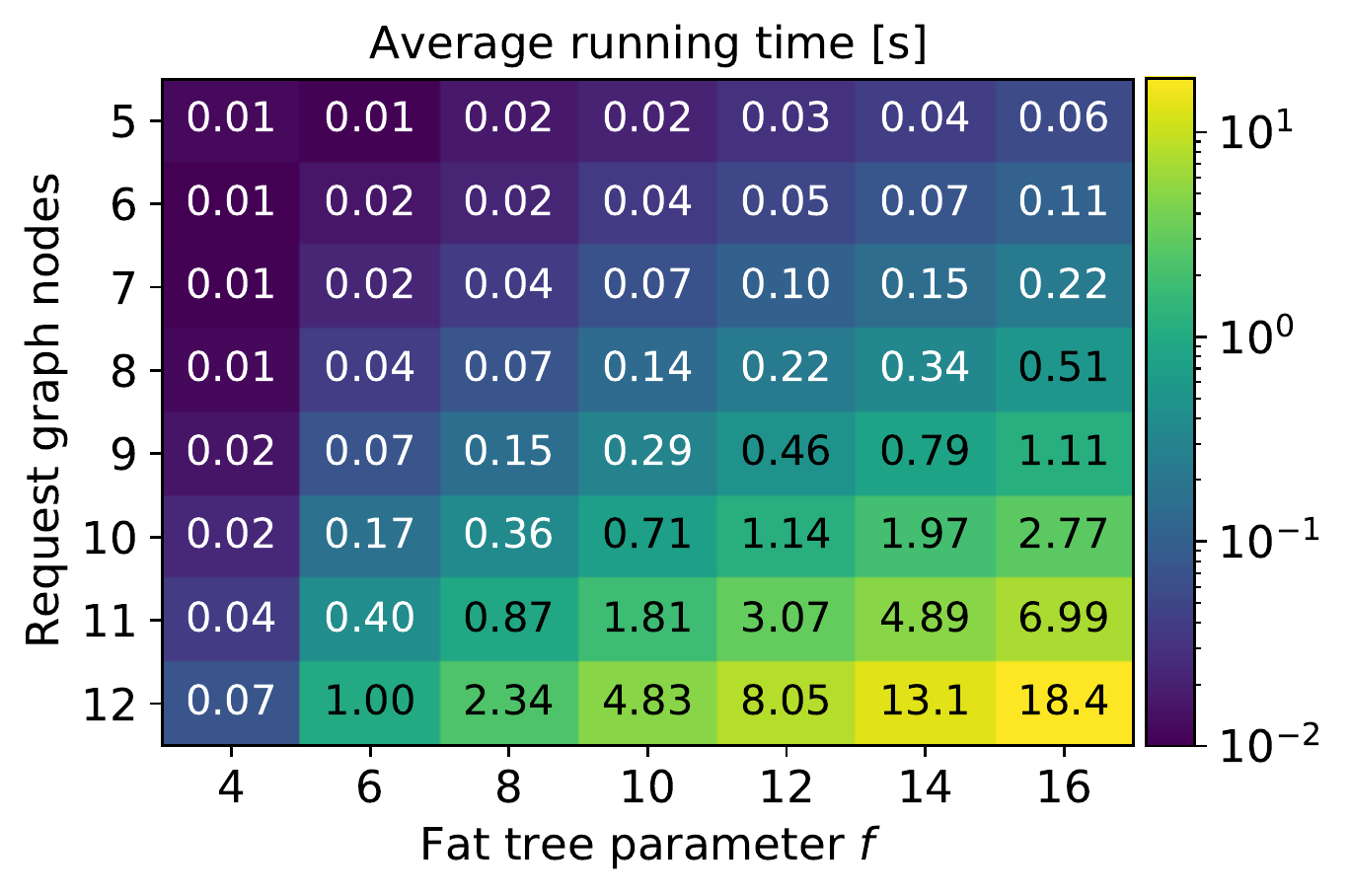}
		\caption{Dynamic program}
	\end{subfigure}
	\begin{subfigure}[t]{0.32\textwidth}
			\centering
		\includegraphics[width=\textwidth]{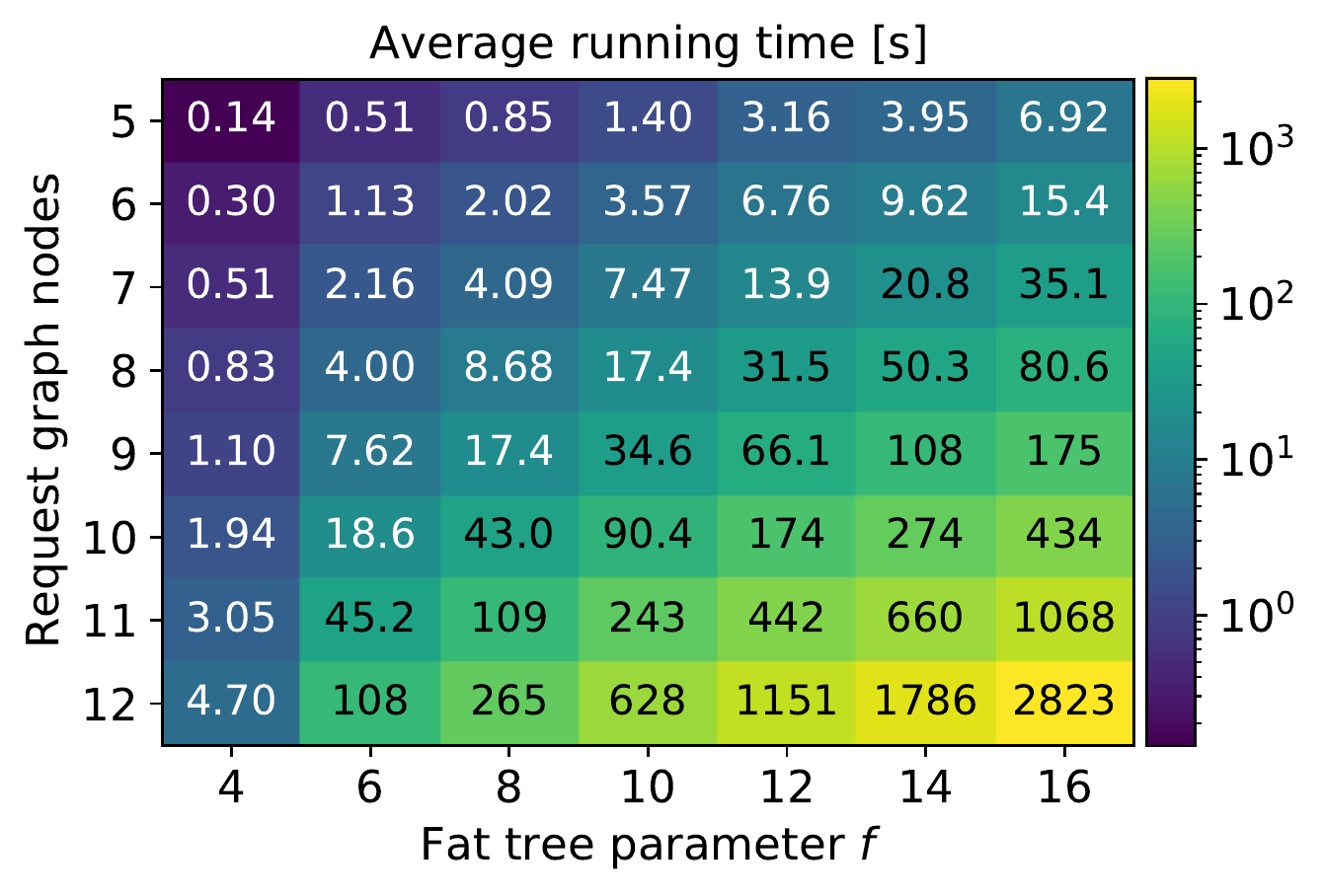}
		\caption{Integer program}
	\end{subfigure}
	\begin{subfigure}[t]{0.32\textwidth}
		\centering
		\includegraphics[width=\textwidth]{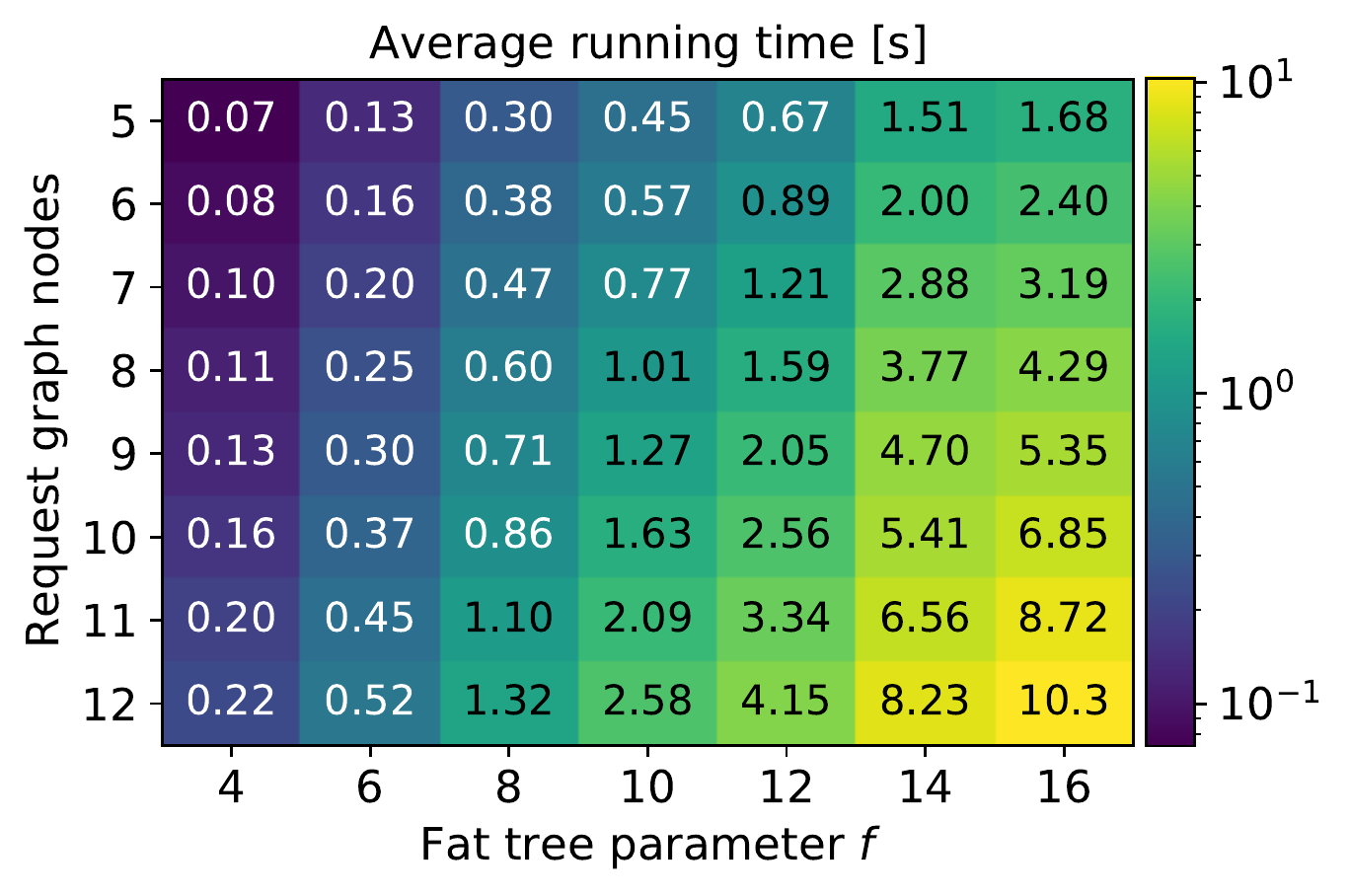}
		\caption{ViNE}
	\end{subfigure}
	\caption{Running time statistics in seconds. Each heatmap cell averages 100 instances of different \ErdosRenyi{} request graphs, 10 for each connection probability~$p \in \{0.1, \dots, 1.0\}$. Recall that for the integer program the time limit is set to $200\times$  the dynamic program's running time. Note the different (logarithmic) z-axes.
} 
	\label{fig:runtime}
	\vspace{-12pt}
\end{figure*}
\fi{}
\begin{figure*}[t]
	\centering
	\begin{subfigure}{0.32\textwidth}
		\includegraphics[width=\textwidth]{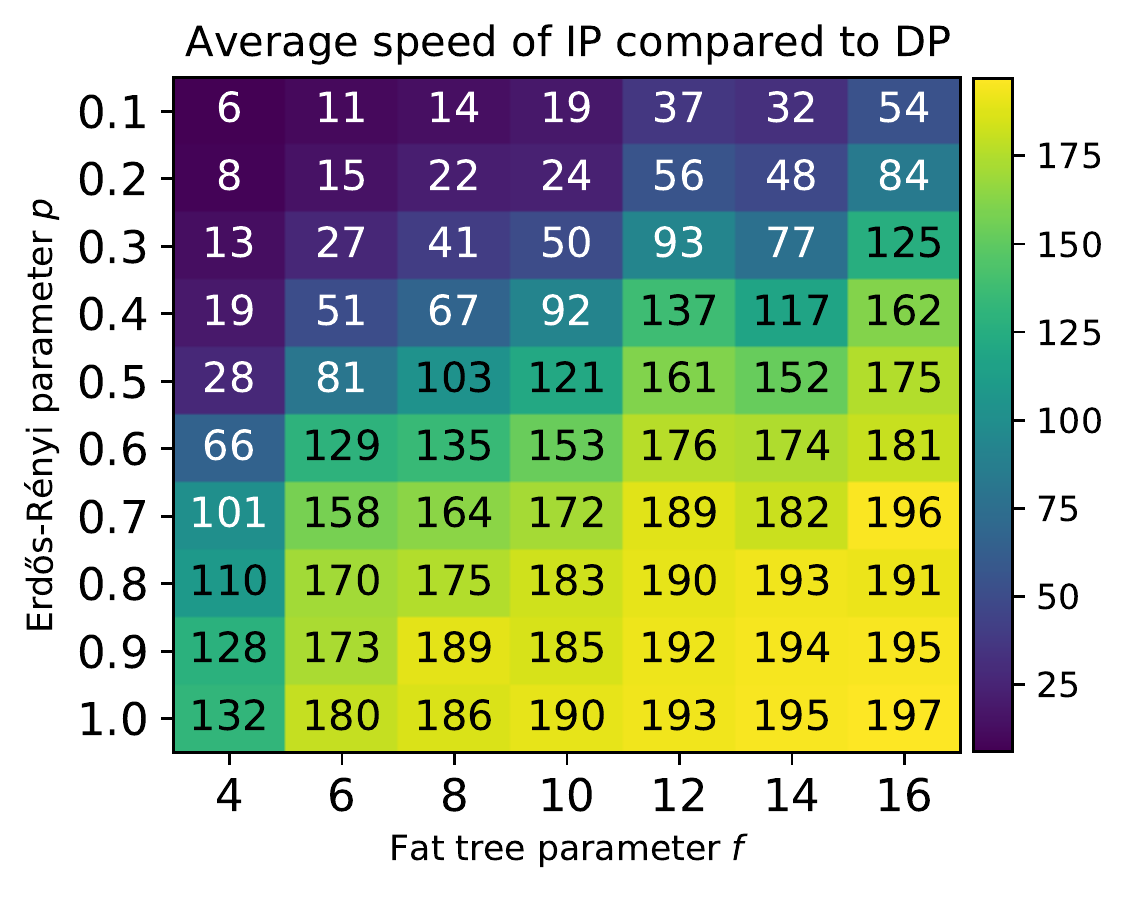}
		\caption{Average ratio of the running time of the integer and the dynamic program. Recall the IP's time limit of $200\times$  the DP's one.}
		\label{fig:ip-speedup}
	\end{subfigure}
	\hfill
	\begin{subfigure}{0.32\textwidth}
		\includegraphics[width=\textwidth]{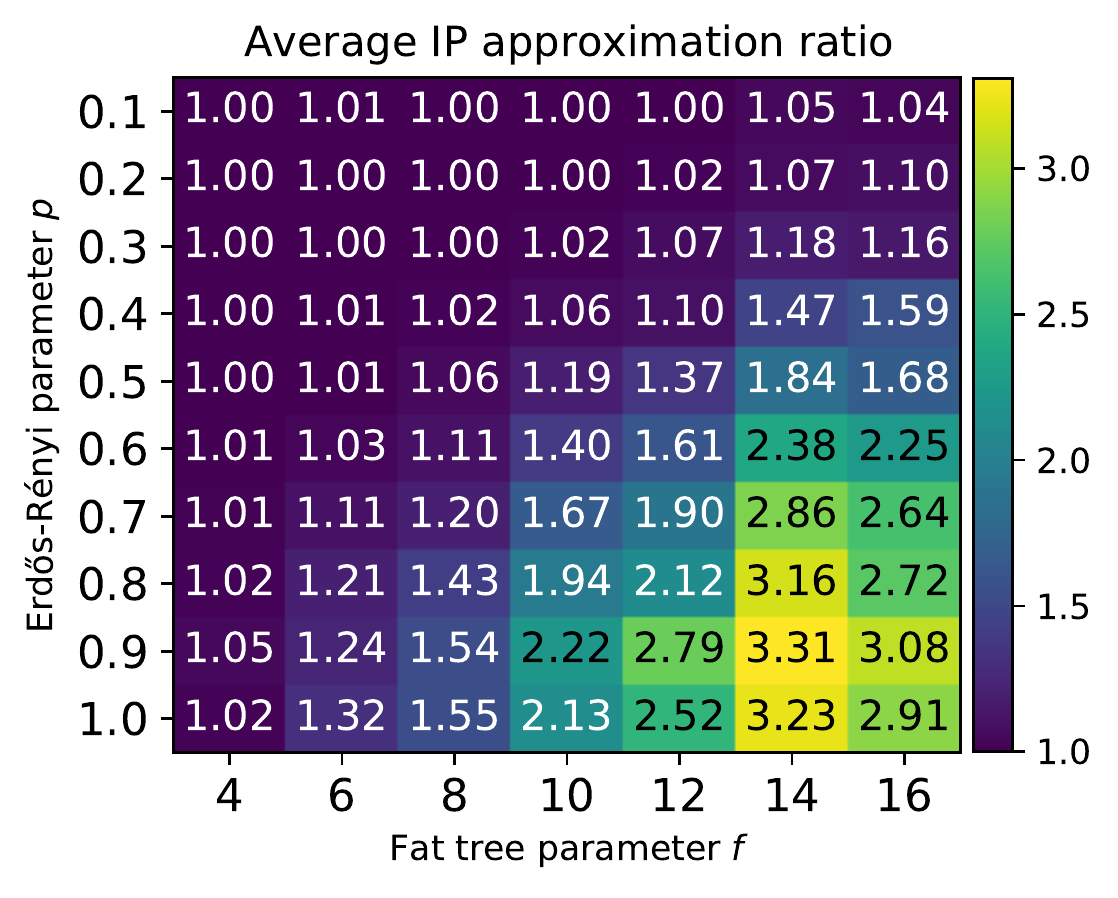}
		\caption{Average ratio of the lowest embedding cost found by the IP within the time limit to the optimal minimum cost found by the DP.}
		\label{fig:ip-approx}
	\end{subfigure}
	\hfill
	\begin{subfigure}{0.32\textwidth}
		\includegraphics[width=0.82\textwidth]{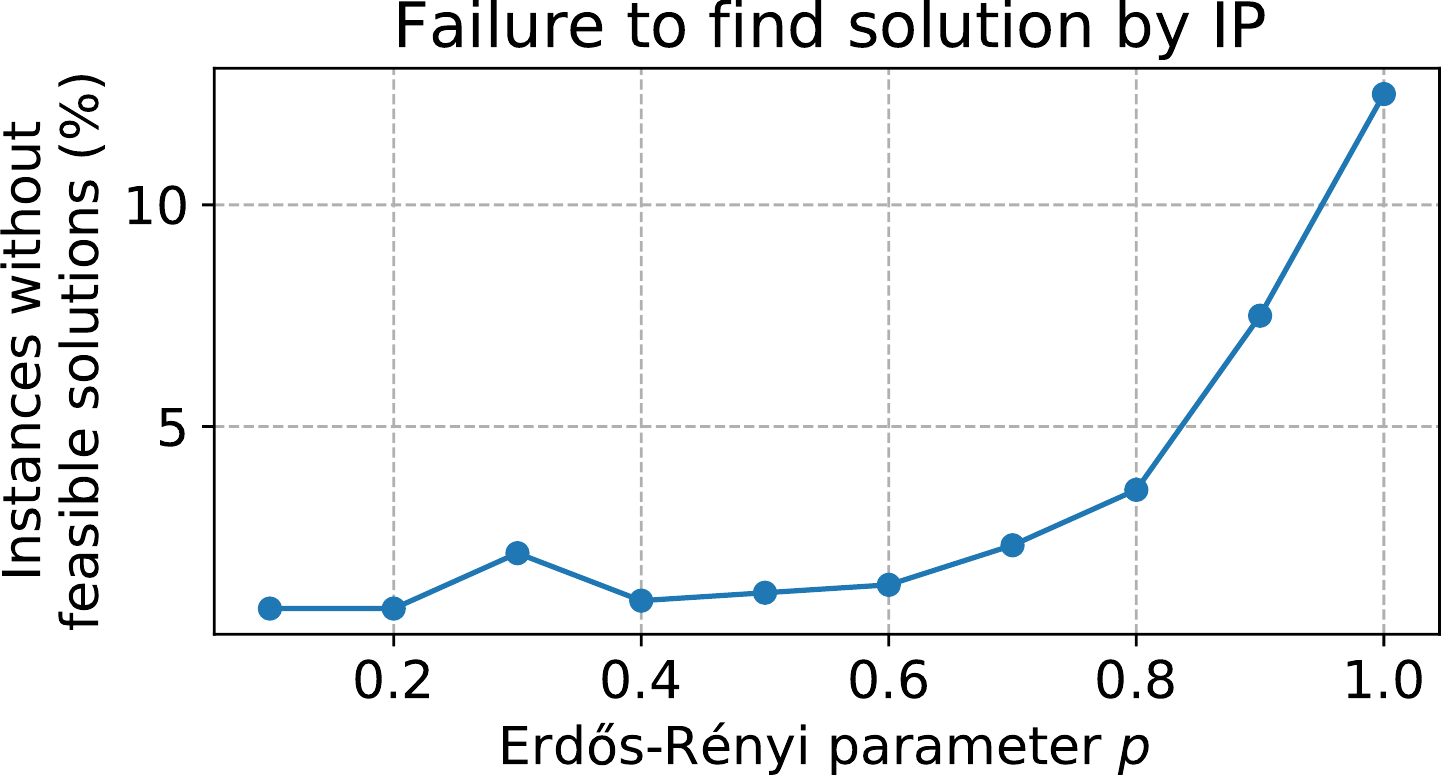}
		\vspace{6pt}
		\includegraphics[width=0.82\textwidth]{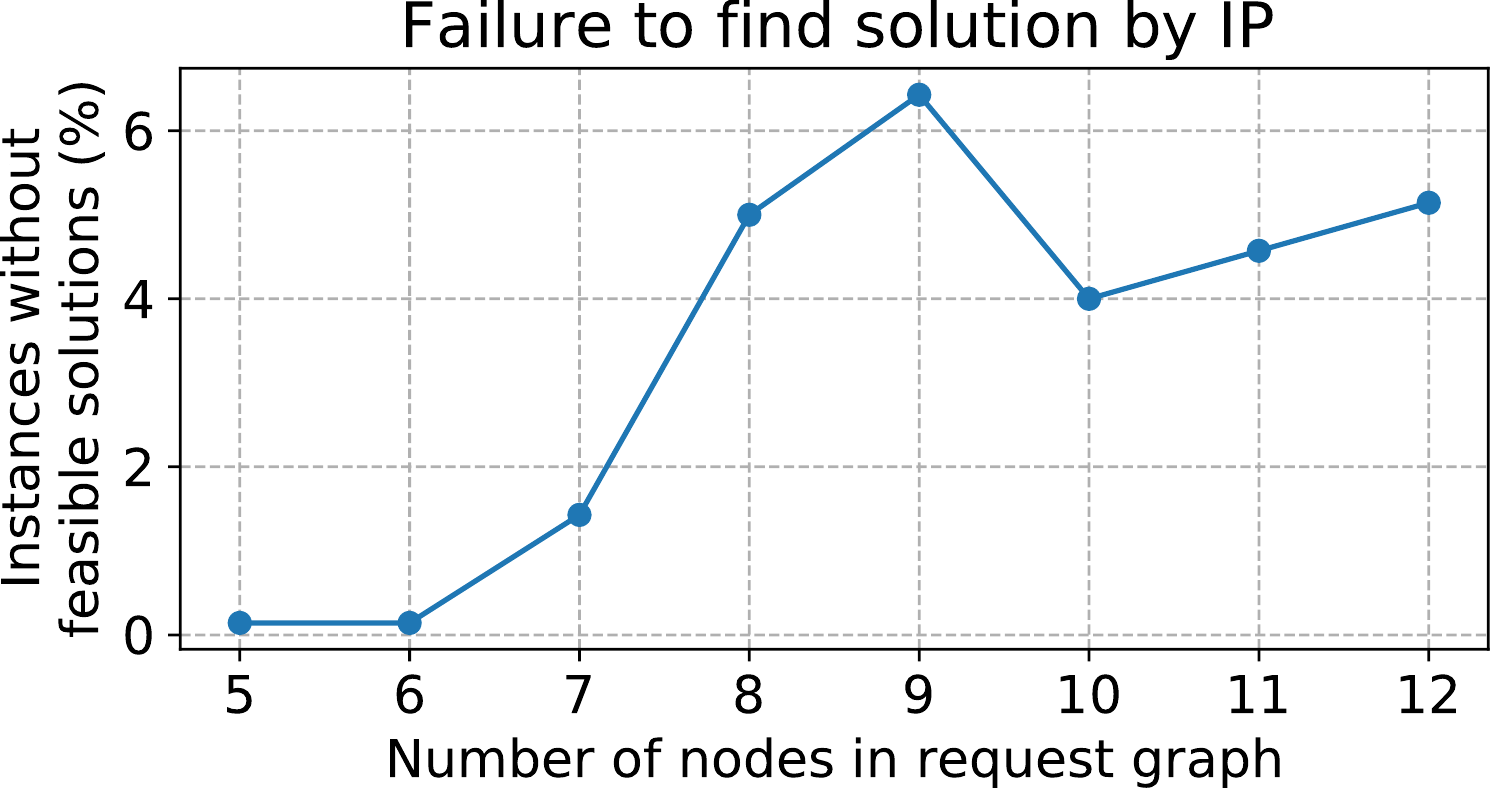}
		\vspace{-8pt}
		\caption{Analysis of the IP's failure to produce a feasible solution within the time limit.}
		\label{fig:ip-failure}
	\end{subfigure}
\vspace{-8pt}
	\caption{Comparison to the IP in terms of running time ratio, approximation ratio, and feasible solutions. Each heatmap cell averages 80 instances of \ErdosRenyi{} request graphs of sizes 5--12.}
	\label{fig:ip}
	\vspace{-8pt}
\end{figure*}
\begin{figure*}[t]
	\centering
	\begin{subfigure}[t]{0.47\textwidth}
		\centering
		\includegraphics[width=0.65\columnwidth]{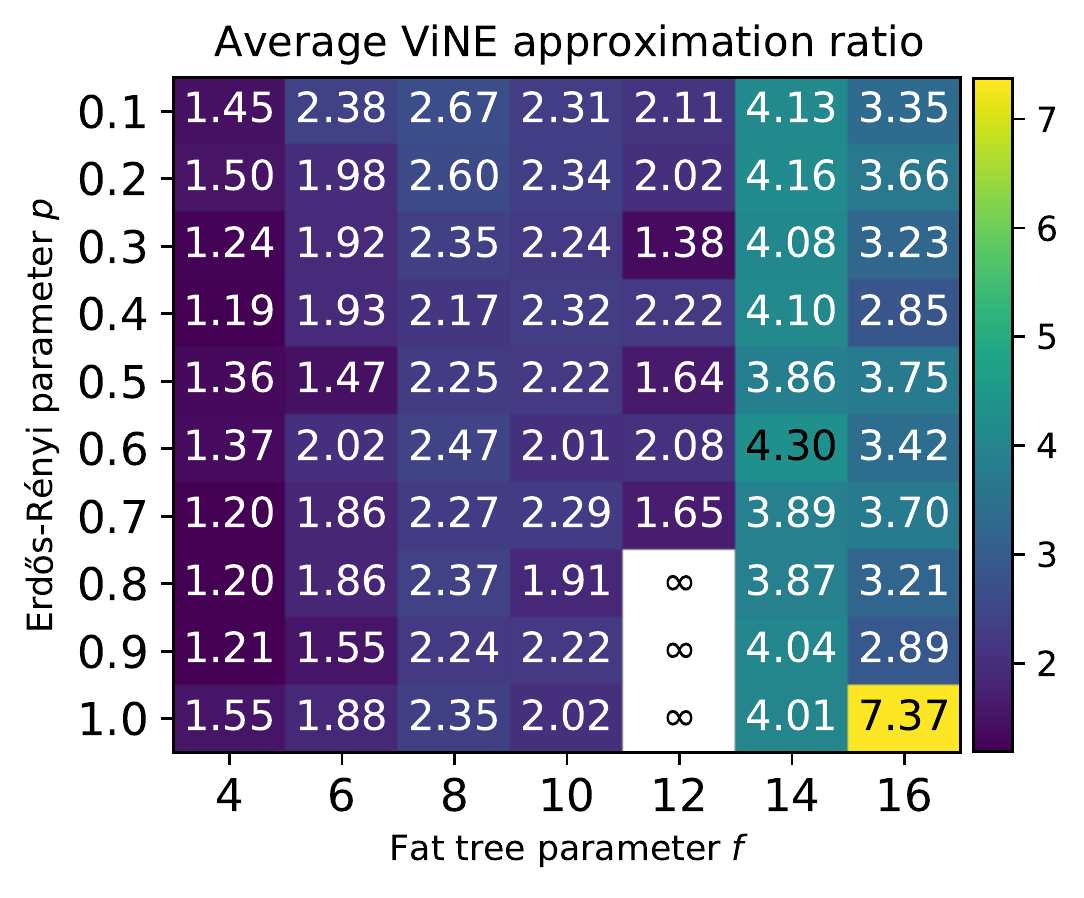}
		\vspace{-0.3cm}
		\caption{Average ratio of the embedding cost found by ViNE to the optimal minimum cost embedding found by the DP.}
		\label{fig:vine-approx}
	\end{subfigure}
	\hfill
	\begin{subfigure}[t]{0.47\textwidth}
		\centering
		\includegraphics[width=0.65\columnwidth]{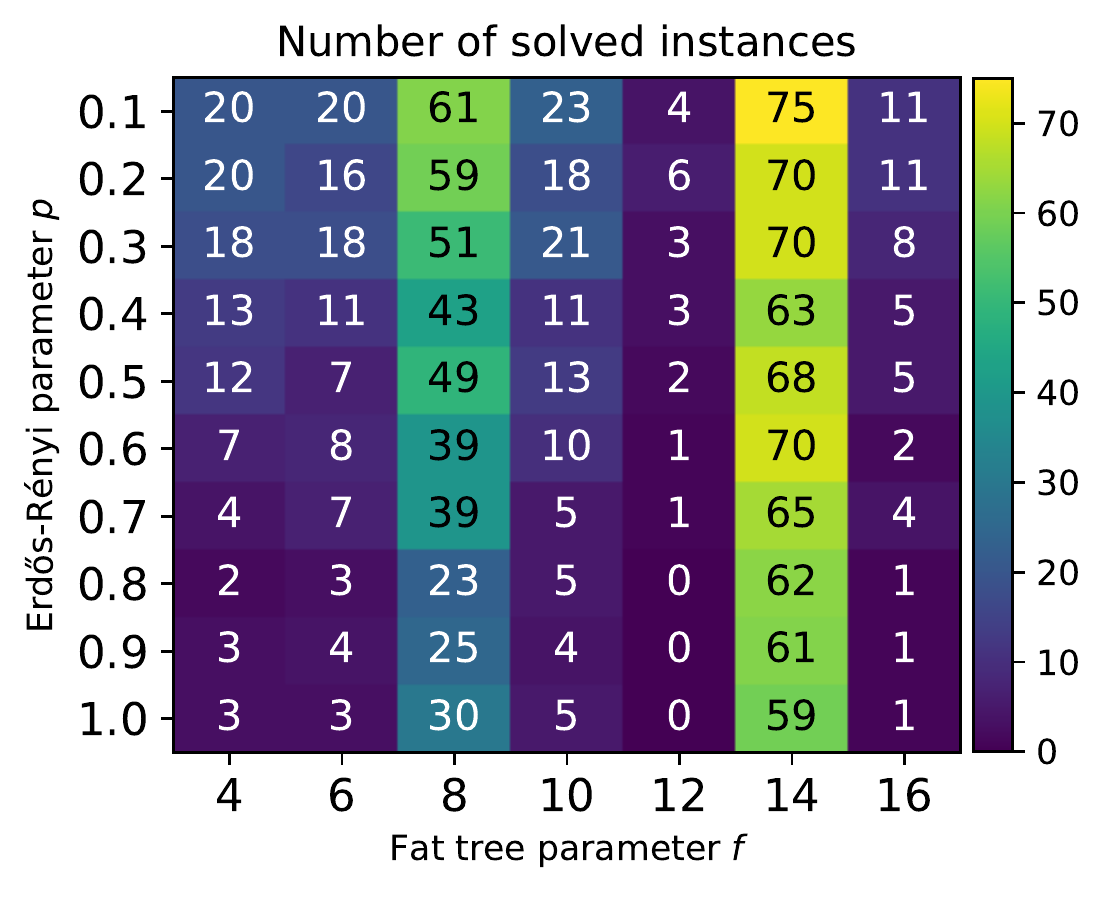}
		\vspace{-0.3cm}
		\caption{Number of instances for which ViNE found a feasible solution out of 80 possible per cell.}
		\label{fig:vine-solved}
	\end{subfigure}
\vspace{-8pt}
	\caption{Comparison to ViNE in terms of approximation quality and number of solved instances. Each heatmap cell averages 80 instances of \ErdosRenyi{} request graphs of sizes 5--12.}
	\label{fig:vine}
	\vspace{-8pt}
\end{figure*}

We compared the implementations on servers equipped with an Intel Xeon W-2125 4-core, 8-thread CPU clocked at 4.0 GHz and 256GB of RAM running Ubuntu 18.04.
In \Cref{fig:runtime}, the running times of our dynamic program (DP) as well as the integer program (IP) and the ViNE heuristic are depicted.
The running time of the DP increases on average by a factor of 2 to 3 with the number of nodes of the request graph.
Notably, this factor lies beneath the proven factor of~$3$~(see~Section~\ref{sec:robust-dp}), as our implementation of the DP skips some redundant computations.
The running time of the IP increases exponentially as well, however due to the enforced time limit, specific growth values could not be gathered. 
The running time of the IP exceeds the one of the DP by at least~$10\times$ for more than~$98.5\%$ of the instances and by at least~$100\times$ for more than~$61.4\%$ of the instances.
The DP is faster than ViNE in~$85\%$ of the instances; the running time of ViNE is better than the one of the DP whenever both the request and the substrate graphs become large.

In \cref{fig:ip-speedup} we further analyze the speedup of the DP over the IP and how it relates to the parameters that control the size of the substrate and density of the requests.
It can be seen that the speedup of the DP increases for larger values of~$p$ and~$f$.
This is likely due to the fact that the number of variables in the IP is~$\mathcal O(|\SV| \cdot (|\VV| + |\VE|))$, while the running time of the DP has exponential dependence only on~$|\VV|$.
The average speedup on instances with large $f$ and $p$ is close to~$200$, meaning that almost always the~$200\times$ time limit was reached.
To better understand the impact of this premature termination, we also report on the (empirical) approximation ratio achieved by the integer program in \cref{fig:ip-approx}:
For instances that the IP could not solve exactly within the time limit, there is a substantial gap in the embedding cost.
Moreover, there were~152 instances (2.7\%) for which the IP could not produce an initial feasible solution within the time limit; note that the DP produced the optimal solution while being~$200\times$ quicker.
\Cref{fig:ip-failure} gives insights into the instances for which this case was encountered.
One can see that the IP struggles to construct solutions for requests with high connectivity~$p$.
The peak number of instances for which the IP did not produce a solution was observed for requests of graph size $9$.
We believe the reason for this to be that the IP spent more time on initialization efforts, such as computing the root linear programming relaxation.

Next, we compare our DP to the ViNE heuristic in terms of approximation quality (see~\cref{fig:vine}).
One can observe that, as opposed to the IP, the approximation ratio of ViNE slightly improves with growing connection probability~$p$ (see~\cref{fig:vine-approx}).
But with growing fat tree parameter~$f$, the solution quality decreases, with ViNE returning a feasible solution only for very few instances (see~\cref{fig:vine-solved}).
Notably, except for fat tree sizes~$f=8$ and~$f=14$, ViNE finds feasible solutions for only~$26\%$ of all instances.
Considering running time and approximation ratio combined we observed that there are~$839$ instances ($15\%$) for which ViNE was faster than the DP.
In~$130$ of those, ViNE found feasible solutions with an average approximation ratio of~$3.64$ and a speedup factor of~$2.64$.

\paragraph{Discussion.}
\label{sec:eval:discussion}

The above results have shown that our dynamic programming algorithm (DP) consistently outperforms the classical integer programming formulation (IP) for~$\minVNEP$ as well as the well-established ViNE heuristic.
While the formulations of the IP and of ViNE may be improved, e.g., by exploiting the tree structure of the substrate, we believe it to be highly unlikely to be possible to close the tremendous performance gap.
Accordingly, we consider the DP a valuable alternative to integer programming based algorithms as well as heuristics based on linear programming relaxations, for request graphs of small or medium size.
For requests on dozens of nodes, a direct application of our DP seems prohibitive, however.
Here, an interesting approach would be to reduce the size of requests to speed up the algorithm heuristically by using clustering techniques.
As already shown by \citet{fuerst-collocation-preclustering-cloudnet-2013}, heuristic and optimal (pre-)clustering schemes to reduce the request size can be beneficial.
Also \citet{mano-vnep-equest-graph-reductions-infocom-2016} discuss request graph reductions and showed that the cost of embedding reduced request graphs only increases linearly while reducing the running times by exponential factors.
We hence consider this an interesting avenue for developing heuristics based on the dynamic program presented in this work; in this way, one may scale beyond medium-sized requests.

\section{Conclusion}
\label{sec:conlusion}

We initiated the study of a parameterized algorithmics approach for the 
fundamental Virtual Network Embedding Problem which lies at the heart of
emerging innovative network architectures that can be tailored
to the application needs.
In particular, we have shown that despite the general hardness of the problem,
efficient and exact algorithms do exist for practically relevant scenarios.
We understand our work as a first step and believe that
it opens several interesting avenues for future research.
In particular, it would be interesting to further investigate the power of polynomial-time data reduction 
through a parameterized lens, also known as kernelization in parameterized algorithmics. 


{ 
{
	\bibliographystyle{plainnat}
	\bibliography{literature}

\begin{thebibliography}{33}
\providecommand{\natexlab}[1]{#1}
\providecommand{\url}[1]{\texttt{#1}}
\expandafter\ifx\csname urlstyle\endcsname\relax
  \providecommand{\doi}[1]{doi: #1}\else
  \providecommand{\doi}{doi: \begingroup \urlstyle{rm}\Url}\fi

\bibitem[Al-Fares et~al.(2008)Al-Fares, Loukissas, and Vahdat]{fattrees}
Mohammad Al-Fares, Alexander Loukissas, and Amin Vahdat.
\newblock A scalable, commodity data center network architecture.
\newblock In \emph{Proceedings of the ACM SIGCOMM Conference on Data
  Communication}, pages 63--74, 2008.
\newblock \doi{10.1145/1402958.1402967}.

\bibitem[Amaldi et~al.(2016)Amaldi, Coniglio, Koster, and
  Tieves]{amaldi2016computational}
Edoardo Amaldi, Stefano Coniglio, Arie M. C.~A. Koster, and Martin Tieves.
\newblock On the computational complexity of the virtual network embedding
  problem.
\newblock \emph{Electronic Notes in Discrete Mathematics}, 52:\penalty0
  213--220, 2016.

\bibitem[Ballani et~al.(2011)Ballani, Costa, Karagiannis, and
  Rowstron]{oktopus}
Hitesh Ballani, Paolo Costa, Thomas Karagiannis, and Ant Rowstron.
\newblock Towards predictable datacenter networks.
\newblock In \emph{Proceedings {ACM} {SIGCOMM} Conference on Applications,
  Technologies, Architectures, and Protocols for Computer Communications},
  pages 242--253, 2011.
\newblock \doi{10.1145/2018436.2018465}.

\bibitem[Bansal et~al.(2015)Bansal, Lee, Nagarajan, and
  Zafer]{bansal2011minimum}
Nikhil Bansal, Kang{-}Won Lee, Viswanath Nagarajan, and Murtaza Zafer.
\newblock Minimum congestion mapping in a cloud.
\newblock \emph{{SIAM} Journal on Computing}, 44\penalty0 (3):\penalty0
  819--843, 2015.
\newblock \doi{10.1137/110845239}.

\bibitem[Chen et~al.(2006)Chen, Huang, Kanj, and Xia]{CHKX06}
Jianer Chen, Xiuzhen Huang, Iyad~A. Kanj, and Ge~Xia.
\newblock Strong computational lower bounds via parameterized complexity.
\newblock \emph{Journal of Computer and System Sciences}, 72\penalty0
  (8):\penalty0 1346--1367, 2006.

\bibitem[Chowdhury et~al.(2012)Chowdhury, Rahman, and Boutaba]{vnep}
Mosharaf Chowdhury, Muntasir~Raihan Rahman, and Raouf Boutaba.
\newblock {ViNEY}ard: Virtual network embedding algorithms with coordinated
  node and link mapping.
\newblock \emph{{IEEE/ACM} Transactions on Networking}, 20\penalty0
  (1):\penalty0 206--219, 2012.
\newblock \doi{10.1109/TNET.2011.2159308}.

\bibitem[Even et~al.(2013)Even, Medina, Schaffrath, and Schmid]{tcs12vnet}
Guy Even, Moti Medina, Gregor Schaffrath, and Stefan Schmid.
\newblock Competitive and deterministic embeddings of virtual networks.
\newblock \emph{Theoretical Computer Science}, 496:\penalty0 184--194, 2013.
\newblock \doi{10.1016/j.tcs.2012.10.036}.

\bibitem[Fellows et~al.(2009)Fellows, Hermelin, Rosamond, and
  Vialette]{fellows2009parameterized}
Michael~R. Fellows, Danny Hermelin, Frances~A. Rosamond, and St{\'e}phane
  Vialette.
\newblock On the parameterized complexity of multiple-interval graph problems.
\newblock \emph{Theoretical Computer Science}, 410\penalty0 (1):\penalty0
  53--61, 2009.
\newblock \doi{10.1016/j.tcs.2008.09.065}.

\bibitem[Fischer et~al.(2013)Fischer, Botero, Beck, de~Meer, and
  Hesselbach]{vnep-survey}
Andreas Fischer, Juan~F. Botero, Michael~T. Beck, Hermann de~Meer, and Xavier
  Hesselbach.
\newblock Virtual network embedding: A survey.
\newblock \emph{IEEE Communications Surveys \& Tutorials}, 15\penalty0
  (4):\penalty0 1888--1906, 2013.
\newblock \doi{10.1109/SURV.2013.013013.00155}.

\bibitem[Flum and Grohe(2006)]{flum2006parameterized}
J{\"o}rg Flum and Martin Grohe.
\newblock \emph{Parameterized Complexity Theory}.
\newblock Springer, 2006.
\newblock \doi{10.1007/3-540-29953-X}.

\bibitem[Fuerst et~al.(2013)Fuerst, Schmid, and
  Feldmann]{fuerst-collocation-preclustering-cloudnet-2013}
Carlo Fuerst, Stefan Schmid, and Anja Feldmann.
\newblock Virtual network embedding with collocation: Benefits and limitations
  of pre-clustering.
\newblock In \emph{Proceedings of the IEEE International Conference on Cloud
  Networking}, pages 91--98, 2013.
\newblock \doi{10.1109/CloudNet.2013.6710562}.

\bibitem[Greenberg et~al.(2011)Greenberg, Hamilton, Jain, Kandula, Kim, Lahiri,
  Maltz, Patel, and Sengupta]{vl2}
Albert~G. Greenberg, James~R. Hamilton, Navendu Jain, Srikanth Kandula,
  Changhoon Kim, Parantap Lahiri, David~A. Maltz, Parveen Patel, and Sudipta
  Sengupta.
\newblock {VL2:} a scalable and flexible data center network.
\newblock \emph{Communications of the {ACM}}, 54\penalty0 (3):\penalty0
  95--104, 2011.
\newblock \doi{10.1145/1897852.1897877}.

\bibitem[Herrera and Botero(2016)]{nfv-sfc-survey}
Juliver~Gil Herrera and Juan~Felipe Botero.
\newblock Resource allocation in {NFV}: A comprehensive survey.
\newblock \emph{IEEE Transactions on Network and Service Management},
  13\penalty0 (3):\penalty0 518--532, 2016.
\newblock \doi{10.1109/TNSM.2016.2598420}.

\bibitem[Impagliazzo et~al.(2001)Impagliazzo, Paturi, and Zane]{IPZ01}
Russell Impagliazzo, Ramamohan Paturi, and Francis Zane.
\newblock Which problems have strongly exponential complexity?
\newblock \emph{Journal of Computer and System Sciences}, 63\penalty0
  (4):\penalty0 512--530, 2001.
\newblock \doi{10.1006/jcss.2001.1774}.

\bibitem[Inf{\"u}hr and Raidl(2011)]{infuhr2011introducing}
Johannes Inf{\"u}hr and G{\"u}nther~R. Raidl.
\newblock Introducing the virtual network mapping problem with delay, routing
  and location constraints.
\newblock In \emph{Proceedings of the International Conference on Network
  Optimization}, pages 105--117, 2011.
\newblock \doi{10.1007/978-3-642-21527-8_14}.

\bibitem[Karp(1972)]{karp1972reducibility}
Richard~M. Karp.
\newblock Reducibility among combinatorial problems.
\newblock In \emph{Complexity of Computer Computations}, pages 85--103.
  Springer, 1972.
\newblock \doi{10.1007/978-1-4684-2001-2_9}.

\bibitem[Lischka and Karl(2009)]{Lischka}
Jens Lischka and Holger Karl.
\newblock A virtual network mapping algorithm based on subgraph isomorphism
  detection.
\newblock In \emph{Proceedings of the ACM Workshop on Virtualized
  Infrastructure Systems and Architectures}, pages 81--88, 2009.
\newblock \doi{10.1145/1592648.1592662}.

\bibitem[Mano et~al.(2020)Mano, Inoue, Mizutani, and
  Akashi]{mano-vnep-equest-graph-reductions-infocom-2016}
Toru Mano, Takeru Inoue, Kimihiro Mizutani, and Osamu Akashi.
\newblock Reducing dense virtual networks for fast embedding.
\newblock \emph{{IEICE} Transactions on Communications}, 103-B\penalty0
  (4):\penalty0 347--362, 2020.

\bibitem[Meng et~al.(2010)Meng, Pappas, and Zhang]{meng2010improving}
Xiaoqiao Meng, Vasileios Pappas, and Li~Zhang.
\newblock Improving the scalability of data center networks with traffic-aware
  virtual machine placement.
\newblock In \emph{Proceedings of the IEEE International Conference on Computer
  Communications}, pages 1154--1162, 2010.
\newblock \doi{10.1109/INFCOM.2010.5461930}.

\bibitem[Mogul and Popa(2012)]{talkabout}
Jeffrey~C. Mogul and Lucian Popa.
\newblock What we talk about when we talk about cloud network performance.
\newblock \emph{ACM SIGCOMM Computer Communication Review}, 42\penalty0
  (5):\penalty0 44--48, 2012.
\newblock \doi{10.1145/2378956.2378964}.

\bibitem[N{\'{e}}meth et~al.(2020)N{\'{e}}meth, Pignolet, Rost, Schmid, and
  Vass]{ifip20vnep}
Bal{\'{a}}zs N{\'{e}}meth, Yvonne~Anne Pignolet, Matthias Rost, Stefan Schmid,
  and Bal{\'{a}}zs Vass.
\newblock Cost-efficient embedding of virtual networks with and without routing
  flexibility.
\newblock In \emph{Proceedings of the {IFIP} Networking Conference}, pages
  476--484, 2020.

\bibitem[Rost and Schmid(2019)]{rostSchmidVNEP_RR_IFIP_18}
Matthias Rost and Stefan Schmid.
\newblock Virtual network embedding approximations: Leveraging randomized
  rounding.
\newblock \emph{{IEEE/ACM} Transactions on Networing}, 27\penalty0
  (5):\penalty0 2071--2084, 2019.
\newblock \doi{10.1109/TNET.2019.2939950}.

\bibitem[Rost and Schmid(2020)]{ton20hard}
Matthias Rost and Stefan Schmid.
\newblock On the hardness and inapproximability of virtual network embeddings.
\newblock \emph{{IEEE/ACM} Transactions on Networking}, 28\penalty0
  (2):\penalty0 791--803, 2020.
\newblock \doi{10.1109/TNET.2020.2975646}.

\bibitem[Rost et~al.(2014)Rost, Schmid, and Feldmann]{rostSchmidFeldmann2014}
Matthias Rost, Stefan Schmid, and Anja Feldmann.
\newblock It's about time: On optimal virtual network embeddings under temporal
  flexibilities.
\newblock In \emph{Proceedings of the International Parallel and Distributed
  Processing Symposium}, pages 17--26, 2014.
\newblock \doi{10.1109/IPDPS.2014.14}.

\bibitem[Rost et~al.(2015)Rost, Fuerst, and Schmid]{ccr15emb}
Matthias Rost, Carlo Fuerst, and Stefan Schmid.
\newblock Beyond the stars: Revisiting virtual cluster embeddings.
\newblock \emph{ACM SIGCOMM Computer Communication Review}, 45\penalty0
  (3):\penalty0 12--18, 2015.
\newblock \doi{10.1145/2805789.2805792}.

\bibitem[Rost et~al.(2019)Rost, D{\"{o}}hne, and Schmid]{ccr19tw}
Matthias Rost, Elias D{\"{o}}hne, and Stefan Schmid.
\newblock Parametrized complexity of virtual network embeddings: dynamic {\&}
  linear programming approximations.
\newblock \emph{{ACM} {SIGCOMM} Computer Communication Review}, 49\penalty0
  (1):\penalty0 3--10, 2019.
\newblock \doi{10.1145/3314212.3314214}.

\bibitem[{Sher Decusatis} et~al.(2012){Sher Decusatis}, {Carranza}, and
  {Decusatis}]{link-aggregation}
C.~J. {Sher Decusatis}, A.~{Carranza}, and C.~M. {Decusatis}.
\newblock Communication within clouds: open standards and proprietary protocols
  for data center networking.
\newblock \emph{IEEE Communications Magazine}, 50\penalty0 (9):\penalty0
  26--33, 2012.
\newblock \doi{10.1109/MCOM.2012.6295708}.

\bibitem[Sonkoly et~al.(2020)Sonkoly, Haja, Németh, Szalay, Czentye, Szabó,
  Ullah, Kim, and Toka]{SHN20}
Balázs Sonkoly, Dávid Haja, Balázs Németh, Márk Szalay, János Czentye,
  Róbert Szabó, Rehmat Ullah, Byung-Seo Kim, and László Toka.
\newblock Scalable edge cloud platforms for iot services.
\newblock \emph{Journal of Network and Computer Applications}, 170:\penalty0
  102785, 2020.
\newblock \doi{10.1016/j.jnca.2020.102785}.

\bibitem[Soualah et~al.(2016)Soualah, Fajjari, Hadji, Aitsaadi, and
  Zeghlache]{soualah-vnep-via-gomory-hu-noms-2016}
Oussama Soualah, Ilhem Fajjari, Makhlouf Hadji, Nadjib Aitsaadi, and Djamal
  Zeghlache.
\newblock A novel virtual network embedding scheme based on {G}omory-{H}u tree
  within cloud's backbone.
\newblock In \emph{Proceedings of the {IEEE/IFIP} Network Operations and
  Management Symposium}, pages 536--542, 2016.
\newblock \doi{10.1109/NOMS.2016.7502855}.

\bibitem[Sun et~al.(2019)Sun, Chen, Yu, Du, and
  Guizani]{Sun-SFC-DC-2019-ACCESS}
Gang Sun, Zhenrong Chen, Hongfang Yu, Xiaojiang Du, and Mohsen Guizani.
\newblock Online parallelized service function chain orchestration in data
  center networks.
\newblock \emph{IEEE Access}, 7:\penalty0 100147--100161, 2019.
\newblock \doi{10.1109/ACCESS.2019.2930295}.

\bibitem[Yuan et~al.(2018)Yuan, Wang, Peng, and Sood]{Yuan-VNEP-DC-2019-ACCESS}
Ying Yuan, Cong Wang, Sancheng Peng, and Keshav Sood.
\newblock Topology-oriented virtual network embedding approach for data
  centers.
\newblock \emph{IEEE Access}, 7:\penalty0 2429--2438, 2018.
\newblock \doi{10.1109/ACCESS.2018.2886270}.

\bibitem[Zhang et~al.(2014)Zhang, Zhani, Jabri, and
  Boutaba]{vdc-Zhang2014Infocom}
Qi~Zhang, Mohamed~Faten Zhani, Maissa Jabri, and Raouf Boutaba.
\newblock Venice: {R}eliable virtual data center embedding in clouds.
\newblock In \emph{Proceedings of the {IEEE} Conference on Computer
  Communications}, pages 289--297, 2014.
\newblock \doi{10.1109/INFOCOM.2014.6847950}.

\bibitem[Zhu and Ammar(2006)]{zhu2006algorithms}
Yong Zhu and Mostafa~H. Ammar.
\newblock Algorithms for assigning substrate network resources to virtual
  network components.
\newblock In \emph{Proceedings of the {IEEE} International Conference on
  Computer Communications}, 2006.
\newblock \doi{10.1109/INFOCOM.2006.322}.

\end{thebibliography}
}
}


\end{document}